\documentclass[a4paper,UKenglish,thm-restate,autoref,cleveref,numberwithinsect,pdfa]{lipics-v2021}

\pdfoutput=1
\hideLIPIcs
\nolinenumbers

\EventEditors{Rupak Majumdar and Alexandra Silva}
\EventNoEds{2}
\EventLongTitle{35th International Conference on Concurrency Theory (CONCUR 2024)}
\EventShortTitle{CONCUR 2024}
\EventAcronym{CONCUR}
\EventYear{2024}
\EventDate{September 9--13, 2024}
\EventLocation{Calgary, Canada}
\EventLogo{}
\SeriesVolume{311}
\ArticleNo{15}

\author{Christel Baier}{Technische Universit\"at Dresden, Germany}{christel.baier@tu-dresden.de}{0000-0002-5321-9343}{}
\author{Jakob Piribauer}{Technische Universit\"at Dresden, Germany; Universit\"at Leipzig, Germany}{jakob.piribauer@tu-dresden.de}{0000-0003-4829-0476}{}
\author{Maximilian Starke}{Technische Universit\"at Dresden, Germany}{}{}{}

\Copyright{Christel Baier, Jakob Piribauer, Maximilian Starke}



\funding{This work was partly funded by the DFG Grant 389792660 as part of TRR 248 (Foundations of Perspicuous Software Systems) and the Cluster of Excellence EXC 2050/1 (CeTI, project ID 390696704, as part of Germany’s Excellence Strategy).
}

\usepackage[T1]{fontenc}
\usepackage{graphicx}

\usepackage{tikz}
\usetikzlibrary{positioning,automata,fit,shapes,calc}
\usepackage{dsfont}
\usepackage{thmtools,thm-restate}
\usepackage{pgfplots}
\usepgfplotslibrary{fillbetween}
\pgfplotsset{width=10cm,compat=1.9}
\usepackage{tabularx}
    \newcolumntype{L}{>{\raggedright\arraybackslash}X}
    \newcolumntype{R}{>{\raggedleft\arraybackslash}X}
\usepackage{wrapfig}

\usepackage{bbold}
\usepackage{mathtools}
\usepackage{todonotes}
\usepackage{enumerate}

\usepackage{caption}
\usepackage{subcaption}

\usepackage{amssymb}

\newcommand{\rawdiaplus}{%
  \begin{tikzpicture}
    \useasboundingbox (-0.7ex, -0.9ex) rectangle (0.7ex, 0.9ex);
    \node (w) at (-0.7ex,0) {};
    \node (e) at (+0.7ex,0) {};
    \node (s) at (0,-0.9ex) {};
    \node (n) at (0,+0.9ex) {};
    \draw (n.center) -- (e.center) -- (s.center) -- (w.center) -- (n.center);
    \draw (n.center) -- (s.center);
    \draw (e.center) -- (w.center);
  \end{tikzpicture}}

\newsavebox{\diamondplusbox}
\savebox{\diamondplusbox}{\rawdiaplus}

\newcommand{\rawdiaminus}{%
  \begin{tikzpicture}
    \useasboundingbox (-0.7ex, -0.9ex) rectangle (0.7ex, 0.9ex);
    \node (w) at (-0.7ex,0) {};
    \node (e) at (+0.7ex,0) {};
    \node (s) at (0,-0.9ex) {};
    \node (n) at (0,+0.9ex) {};
    \draw (n.center) -- (e.center) -- (s.center) -- (w.center) -- (n.center);
    \draw (e.center) -- (w.center);
  \end{tikzpicture}}

\newsavebox{\diamondminusbox}
\savebox{\diamondminusbox}{\rawdiaminus}

\newcommand{\MAD}{\mathbb{MAD}}
\newcommand{\SV}{\mathbb{SV}}
\newcommand{\SVPE}{\mathbb{SVPE}}
\newcommand{\SMAD}{\mathbb{SMAD}}

\newcommand{\MADPE}{\mathbb{MADPE}}

\newcommand{\Var}{\mathbb{V}}

\newcommand{\ERMin}{ERMin}
\newcommand{\ERMax}{ERMax}
\newcommand{\TBP}{\mathit{TBP}}

\newcommand{\cE}{\mathcal{E}}

\newcommand{\cM}{\mathcal{M}}
\newcommand{\cN}{\mathcal{N}}

\newcommand{\eqdef}{\ensuremath{\stackrel{\text{\tiny def}}{=}}}

\renewcommand{\Pr}{\mathrm{Pr}}

\newcommand{\sinit}{s_{\mathit{\scriptscriptstyle init}}}

\newcommand{\Act}{\mathit{Act}}

\newcommand{\act}{\alpha}

\newcommand{\fpath}{\pi}

\newcommand{\last}{\mathit{last}}

\newcommand{\Cyl}{\mathit{Cyl}}

\newcommand{\Sched}{\mathit{Sched}}

\newcommand{\sched}{\mathfrak{S}}
\newcommand{\tsched}{\mathfrak{T}}

\newcommand{\rew}{\mathit{rew}}

\newcommand{\goal}{\mathit{goal}}

\newcommand{\Rational}{\mathbb{Q}}

\newcommand{\CiteAppendix}[1]{}

\newtheorem{assumption}{Assumption}{\bfseries}{\itshape}

\begin{document}
\title{Risk-averse optimization of total rewards in Markovian models using deviation measures}
\titlerunning{Risk-averse optimization of total rewards in Markovian models}

\authorrunning{Christel Baier, Jakob Piribauer, Maximilian Starke}
\maketitle              

\begin{abstract}
This paper addresses objectives tailored to the  \emph{risk-averse} optimization of accumulated rewards in Markov decision processes (MDPs).
The studied objectives require maximizing the expected value of the accumulated rewards minus a penalty factor times a deviation measure of the resulting distribution of rewards.
Using the variance in this penalty mechanism leads to the variance-penalized expectation (VPE) for which it is known that optimal schedulers have to minimize future expected rewards when a high amount of rewards has been accumulated.
This behavior is undesirable as risk-averse behavior should keep the probability of particularly low outcomes low, but not discourage the accumulation of additional rewards on already good executions.

The paper investigates the semi-variance, which only takes outcomes below the expected value into account, the mean absolute deviation (MAD), and the semi-MAD as alternative deviation measures.
Furthermore, a penalty mechanism that penalizes outcomes below a fixed threshold  is studied.
For all of these objectives, the properties of optimal schedulers are specified and in particular the question whether these objectives overcome the problem observed for the VPE is answered. 
Further, the resulting algorithmic problems on MDPs and Markov chains are investigated.
\end{abstract}

\section{Introduction}

Markov decision processes (MDPs) are a prominent  model for systems 
whose behavior is subject to  \emph{non-determinism} and  \emph{probabilism}. 
Non-deterministic behavior might arise, e.g., if a system is employed in an unknown environment, 
can be controlled by a user, or  works concurrently.
On the other hand, if, e.g., sufficiently much data on the failure of  components is available or randomized algorithms make use of randomization explicitly, 
it is reasonable to model these aspects of the system as probabilistic.

In order to model quantitative aspects of a system, such as energy consumption, execution time, or utility, MDPs are often equipped  with a  \emph{reward function} that specifies how much  reward is received in each step of an execution. 
A typical task is then to resolve the non-deterministic choices by specifying a \emph{scheduler}, a.k.a. \emph{policy}, such that the expected value of the total accumulated  reward is maximal (or minimal). 
In verification, such optimization problems naturally occur when investigating the worst- or best-case expected value of the accumulated  reward where worst- and best-case range over all resolutions of the non-deterministic choices.
If additionally a target state has to be reached almost surely, this problem is known as the \emph{stochastic shortest path problem} \cite{Bertsekas1991AnAO,deAlf99}.

\paragraph*{Risk-averse optimization.}
If the objective is the maximization of  the expected value of the accumulated rewards, all other aspects of the probability distribution of accumulated rewards are disregarded.
This might lead to undesirable behavior as the optimal scheduler might receive low rewards with high probability as long as the expected value is optimal. In many situations, however, 
a slightly lower expected reward is preferable if it is obtained by a more ``stable'' behavior in which the risk of encountering low rewards is reduced. E.g., in a traffic control scenarios, it might be  important to reduce the risk of congestions while ensuring a reasonable average throughput instead of solely optimizing the average throughput. 

In order to define objectives incentivizing such risk-averse behavior, it is worth taking a look
at finance and in particular portfolio optimization.
Here,  Markowitz proclaimed  that a portfolio of financial positions should be chosen such that it is Pareto optimal with respect to the expected return and the variance of the return \cite{markowitz52}. 
One way extensively studied in finance to obtain Pareto optimal portfolios is to maximize the \emph{variance-penalized expectation (VPE)}, which is the expected value minus a penalty factor $\lambda$ times the variance. 
The parameter $\lambda$ can be used to obtain different levels of risk-aversion.

Besides the variance, further deviation measures have been investigated to reduce risk in portfolio optimization:
The use of the \emph{semi-variance}, which -- in contrast to the variance -- only takes the deviation of outcomes below the expected value into account, as a penalty mechanism  has been 
introduced in this context by Markowitz  \cite{markowitz59}. 
Furthermore, instead of considering quadratic deviations from the expected value as in the case of variance and semi-variance, the \emph{mean absolute deviation (MAD)} can be used to obtain  the MAD-penalized expectation (MADPE)
studied for portfolio optimization in \cite{Konno1991}.
The MAD measures the expected absolute deviation from the expected value. 
 
 In this paper, we investigate these different deviation measure based penalty mechanisms in the context of the maximization of rewards in MDPs.

\paragraph*{Variance-penalized expectation in MDPs (VPE).}
Recently, the maximization of the VPE of accumulated rewards in MDPs was studied in \cite{Piribauer2022TheVS}:
On the positive side, it is shown  that optimal schedulers for the VPE can be chosen to be deterministic finite-memory schedulers. 
Nevertheless, the optimization of the VPE is shown to be computationally hard: The threshold problem whether the optimal VPE exceeds a given threshold $\vartheta$ is EXPTIME-hard. An optimal scheduler can be computed in exponential space.

A main drawback of the VPE, however, is of conceptual nature: In  \cite{Piribauer2022TheVS}, it is shown that VPE-optimal schedulers have to \emph{minimize} the future expected rewards as soon as a high amount of rewards (above a computable bound $B$) has been accumulated. We call such schedulers \emph{eventually reward-minimizing schedulers (\ERMin-schedulers)}. 
Intuitively, the reason is that a further accumulation of additional rewards after a high amount of rewards has already been accumulated has a stronger effect on the variance than on the expected value due to the quadratic nature of the variance. Conceptually, this can be considered to be a flaw in the use of the VPE as an objective  to yield  risk-averse behavior.

The desired behaviour a suitable objective should induce is that a scheduler  achieves a high expected accumulated reward, while keeping the probability of particularly bad outcomes low. 
Improving on already good outcomes should not have a negative effect. So, we want optimal schedulers to be \emph{eventually reward-maximizing (\ERMax-schedulers)}, i.e., that they maximize the  expected reward once the accumulated reward exceeds some bound $B$. 

\paragraph*{Deviation-measure-penalized expectation.}
Towards this goal, 
we investigate objectives in the spirit of the VPE, which are of the form 
$
\mathbb{E}^\sched(\rew) - \lambda \mathbb{DEV}^\sched(\rew)
$
where a penalty factor $\lambda$ times a deviation measure $\mathbb{DEV}^\sched(\rew)$ of the probability distribution of accumulated rewards under a scheduler $\sched$ is subtracted from the expected accumulated reward $\mathbb{E}^\sched (\rew)$.

The first  deviation measure we investigate is the MAD. In contrast to the variance, the contribution of an outcome to the MAD only grows linearly with its distance to the expected value.
For the MAD and the variance, we also study one-sided variants in which only outcomes below the expected value are considered: The semi-MAD (SMAD) and semi-variance quantify the average absolute or squared deviation below the expected value by assigning deviation $0$ to all outcomes above the expected value.
 Finally, we investigate a simpler alternative to the MADPE: Instead of measuring the deviation from the expected value of accumulated rewards, which itself depends on the chosen scheduler,
we consider a threshold-based penalized expectation (TBPE), where outcomes below a threshold $t$ that can be chosen externally  are penalized either linearly or according to more complicated functions.

\begin{table}[t]
\begin{tabularx}{\textwidth}{L| L| L| L}
&hardness of threshold problem& computation of optimum & optimal schedulers \\
\hline
VPE \cite{Piribauer2022TheVS} & EXPTIME-hard;  in P for Markov chains & in exponential space & deterministic, finite-memory \ERMin-schedulers \\
\hline
SVPE & -  & -  & randomized, \ERMin-schedulers can be necessary \\
\hline
MADPE ($\lambda \leq 1/2$), SMADPE ($\lambda \leq 1$) & PP-hard for acyclic Markov chains  & quadratic program of exponential size & randomized, finite-memory \ERMax-schedulers \\
\hline
MADPE ($\lambda > 1/2$), SMADPE ($\lambda > 1$)  & - & - & randomized, \ERMin-schedulers can be necessary\\
\hline
TBPE &PP-hard for acyclic Markov chains  & in pseudo-polynomial time & deterministic, finite-memory \ERMax-schedulers \\
\hline
\end{tabularx}
\vspace{6pt}
\caption{Overview of the complexity results and the types of schedulers needed for the optimization of the studied objectives and the VPE. The entries ``-'' indicate that the problem was not studied further as the scheduler needed for the optimization are the undersirable \ERMin-schedulers.}
\label{tbl:overview}
\vspace{-24pt}
\end{table}

\paragraph*{Contributions.}
The main contributions, also summarized  in Table \ref{tbl:overview}, are as follows. 
\begin{itemize}
\item
We show that  optimal schedulers for the MADPE can be chosen to be \ERMax-schedulers, as desired, if the  risk-aversion parameter $\lambda$ is sufficiently small, i.e. if $\lambda \leq 1/2$.
This bound on the parameter is shown to be tight. 
Furthermore, we show that randomized schedulers are necessary for the optimization.

We formulate the optimization problem as a quadratic program and obtain a EXPSPACE-upper complexity bound for the threshold problem for the MADPE.
On the other hand, we show that already in acyclic Markov chains the threshold problems for the MADPE and the MAD are PP-hard 
under polynomial-time Turing reductions.

As the semi-MAD is  always half of the MAD, the results transfer to the semi-MADPE.
\item
We 
investigate the semivariance-penalized expectation (SVPE) and show -- somewhat surprisingly -- that,  for any risk-aversion parameter $\lambda$, there are MDPs in which optimal schedulers are \ERMin-schedulers. Hence, the SVPE as objective does not overcome the undesirable effects observed for the VPE. Furthermore, we show that, in contrast to the VPE, randomization is necessary for the optimization of the SVPE.
\item
We show that the TBPE can be optimized in pseudo-polynomial time and that
  deciding if the TBPE exceeds a  bound 
for  linear penalty functions even in acylic Markov chains is PP-hard under polynomial-time Turing reductions.
\end{itemize}
As a proof-of-concept, we  analyze our algorithms for the optimization of the MADPE and for the TBPE in a small series of experiments.

\paragraph*{Related work.}
The above mentioned work on the VPE for accumulated rewards in MDPs \cite{Piribauer2022TheVS} is the closest related work to our paper.
Earlier work on the VPE in MDPs addressed the finite-horizon  setting with terminal rewards  \cite{collins1997finite} or applied the notion to mean payoff and discounted rewards \cite{filar1989variance}.
Further, \cite{xia2020risk} presents a policy iteration algorithm converging against \emph{local} optima for a similar measure.
The computation of the variance of accumulated rewards has been studied in Markov chains \cite{verhoeff2004reward} and in MDPs \cite{mandl1971variance,MannorTsitsiklis2011}.
In \cite{DBLP:journals/jcss/BrazdilCFK17}, the satisfiability of constraints on the expected mean payoff in conjunction with constraints on the variance or related notions such as a local variability are studied for MDPs.

For MDPs, the 
SVPE of random variables defined in terms of the stationary distribution has been studied via the use of reinforcement learning algorithms \cite{DBLP:journals/jair/MaMXZ22}. 
Conceptually and methodologically this work is nevertheless not closely related to our work. 
We are not aware of investigations of the MADPE on MDPs.

Furthermore, several approaches to formalize various other {risk-averse} optimization problems for accumulated rewards in MDPs have been proposed and studied in the literature.
This includes the computation of worst- or best-case quantiles  \cite{UB13,DBLP:conf/nfm/BaierDDKK14,HaaseKiefer15,DBLP:journals/fmsd/RandourRS17}, also called \emph{values-at-risk}:
Given a probability $p$, quantiles on the accumulated rewards are the best bound $C$ such that the accumulated rewards stays below $C$ with probability at  most $p$ under all or under some scheduler. 
While quantiles still disregard the distribution below, 
the \emph{conditional value-at-risk} and the \emph{entropic value-at-risk} are more involved measures that quantify how far the probability mass of  the tail of the probability distribution lies below a given quantile.
In the context of risk-averse optimization in MDPs, these measures  have been studied in \cite{kretinsky2018} and  \cite{ahmadi2021}.
A further approach, the \emph{entropic risk} measure, reweighs outcomes by an exponential utility function. Optimizing this entropic risk measure leads to schedulers that tend to still achieve a high expected value while keeping the probability of low outcomes small. The entropic risk measure applied to accumulated rewards have been studied in \cite{mfcs2023} for stochastic games that extend MDPs with an adversarial  player.

\section{Preliminaries}
\label{sec:prelim}

\paragraph*{Notations for Markov decision processes.}
A \emph{Markov decision process} (MDP) is a tuple $\mathcal{M} = (S,\Act,P,\sinit,\rew)$
where $S$ is a finite set of states,
$\Act$ a finite set of actions,
$P \colon S \times \Act \times S \to [0,1] \cap \Rational$  the
transition probability function,
$\sinit \in S$ the initial state,
and
$\rew \colon S \times \Act \to \mathbb{N}$ the reward function. Note that we only allow non-negative rewards and that rational rewards can be transformed to integral rewards by multiplying all rewards with the least common multiple of all denominators of the rational rewards.
We require that
$\sum_{t\in S}P(s,\act,t) \in \{0,1\}$
for all $(s,\alpha)\in S\times \Act$.
We say that action $\alpha$ is enabled in state $s$ iff $\sum_{t\in S}P(s,\act,t) =1$ and denote the set of all actions that are enabled in state $s$ by $\Act(s)$. If $\Act(s)=\emptyset$, we say that $s$ is a \emph{trap} state.
The paths of $\cM$ are finite or
infinite sequences $s_0 \, \act_0 \, s_1 \, \act_1  \ldots$
where states and actions alternate such that
$P(s_i,\act_i,s_{i+1}) >0$ for all $i\geq0$.
For $\fpath =
    s_0 \, \act_0 \, s_1 \, \act_1 \,  \ldots \act_{k-1} \, s_k$,
$\rew(\fpath)=
    \rew(s_0,\act_0) + \ldots + \rew(s_{k-1},\act_{k-1})$ -- and analogously for infinite paths --
denotes the accumulated reward of $\pi$,
$P(\fpath) =
    P(s_0,\act_0,s_1)
    \cdot \ldots \cdot P(s_{k-1},\act_{k-1},s_k)$
its probability, and
$\last(\fpath)=s_k$ its last state. A path is called \emph{maximal} if it is infinite or ends in the trap state $\goal$.
The \emph{size} of $\cM$
is the sum of the number of states
plus the total sum of the logarithmic lengths of the non-zero
probability values
$P(s,\alpha,s')$ as fractions of co-prime integers and the weight values $\rew(s,\alpha)$.

A \emph{Markov chain} is an MDP in which the set of actions is a singleton. In this case, we can drop the set of actions  and consider a Markov chain as a tuple $\cM=(S,P,\sinit, \rew)$ where 
$P$ now is a function from $S\times S$ to $[0,1]$ and $\rew$ a function from $S$ to $\mathbb{N}$.

An \emph{end component} of $\cM$ is a strongly connected sub-MDP formalized by a subset $S^\prime\subseteq S$ of states and a non-empty subset $\mathfrak{A}(s)\subseteq \Act(s)$  for each state $s\in S^\prime$ such that for each $s\in S^\prime$, $t\in S$ and $\alpha\in \mathfrak{A}(s)$ with $P(s,\alpha,t)>0$, we have $t\in S^\prime$ and such that in the resulting sub-MDP all states are reachable from each other.
An end-component is a $0$-end-component if
it only contains state-action-pairs with reward $0$.
%


\paragraph*{Scheduler.}
A \emph{scheduler} for $\cM$
is a function $\sched$ that assigns to each non-maximal path $\fpath$
a probability distribution over $\Act(\last(\fpath))$.
If the choice of a scheduler $\sched$ depends only on the current state, i.e., if $\sched(\fpath)=\sched(\fpath^\prime)$ for all non-maximal paths $\fpath$ and $\fpath^\prime$ with $\last(\fpath)=\last(\fpath^\prime)$,
we say that $\sched$ is \emph{memoryless} and also view it as functions mapping states $s\in S$ to probability distributions over $\Act(s)$.
A scheduler $\sched$ that satisfies $\sched(\fpath)=\sched(\fpath^\prime)$ for all pairs of finite paths $\fpath$ and $\fpath^\prime$ with $\last(\fpath)=\last(\fpath^\prime)$ and $\rew(\fpath)=\rew(\fpath^\prime)$ is called \emph{reward-based} and can be viewed as a function from state-reward pairs $S\times \mathbb{N}$ to probability distributions over actions.
If there is a finite set $X$ of memory modes and a memory update function $U:S\times \Act \times S \times X \to X$ such that the choice of $\sched$ only depends on the current state after a finite path and the memory mode obtained from updating the memory mode according to $U$ in each step, we say that $\sched$ is a finite-memory scheduler.
A scheduler $\sched$ is called deterministic if $\sched(\fpath)$ is a Dirac distribution
for each path $\fpath$ in which case we also view the scheduler as a mapping to actions in $\Act(\last(\fpath))$.

\paragraph*{Probability measure.}
We write $\Pr^{\sched}_{\cM,s}$ 
to denote the probability measure induced by a scheduler $\sched$ and a state $s$ of an MDP $\cM$.
It is defined on the $\sigma$-algebra generated by the {cylinder sets} $\Cyl(\pi)$  of all maximal extensions of a finite path  $\pi =
    s_0 \, \act_0 \, s_1 \, \act_1 \,  \ldots \act_{k-1} \, s_k$ with $s_0=s$ by assigning  to $\Cyl(\pi)$ the probability that $\pi$ is realized under $\sched$, which is
   $ \sched(s_0)(\act_0) \cdot P(s_0,\act_0,s_1) \cdot 
     \ldots \cdot \sched(s_0\act_0 \dots s_{k-1})(\act_{k-1}) \cdot P(s_{k-1},\act_{k-1},s_k)$.
For a set of states $T$, we use $\Diamond T$ to denote the event that a state in $T$ is reached.
For details, see \cite{Puterman}.

For a random variable $X$ that is defined on (some of the) maximal paths in $\cM$, we denote the expected value of $X$ under the probability measure induced by a scheduler $\sched$ and state $s$ by $\mathbb{E}^{\sched}_{\cM,s}(X)$.
We define
$\mathbb{E}^{\min}_{\cM,s}(X) = \inf_{\sched} \mathbb{E}^{\sched}_{\cM,s}(X)$
and
$\mathbb{E}^{\max}_{\cM,s}(X) = \sup_{\sched} \mathbb{E}^{\sched}_{\cM,s}(X)$
where $\sched$ ranges over all schedulers for $\cM$ under which $X$ is defined almost surely.
The variance of $X$ under the probability measure determined by $\sched$ and $s$ in $\cM$ is denoted by $\Var^{\sched}_{\cM,s}(X)$ and defined by
$
    \Var^{\sched}_{\cM,s}(X)\eqdef\mathbb{E}^{\sched}_{\cM,s}((X-\mathbb{E}^{\sched}_{\cM,s}(X))^2)=\mathbb{E}^{\sched}_{\cM,s}(X^2) -\mathbb{E}^{\sched}_{\cM,s}(X)^2.
$
Furthermore, for a measurable set of paths $\psi$ with positive probability, $\mathbb{E}^{\sched}_{\cM,s}(X|\psi)$ denotes the conditional expectation of $X$ under $\psi$.
If $s=\sinit$, we sometimes drop the subscript $s$.

\paragraph*{Accumulated rewards.}
Given an MDP $\cM=(S,\Act,P,\sinit,\rew)$,  the total accumulated reward  is given by the extension of the function $\rew$ to maximal paths.
We can check whether $\mathbb{E}^{\max}_{\cM}(\rew) = \infty$ by checking whether all (maximal) end components are $0$-end components in polynomial time \cite{deAlf99}.
For our purposes, only MDPs $\cM$ with $\mathbb{E}^{\max}_{\cM}(\rew) < \infty$ are interesting. In these MDPs, we can collapse all end components  $\cE$, which are all $0$-end components, to single states $s_\cE$ while adding a transition with reward $0$ to a new trap state. This does not affect the possible distributions of the random variable $\rew$ that can be realized by a scheduler \cite{deAlf99}.
Furthermore, the behavior of the MDP starting from a state $s$ with $\mathbb{E}^{\max}_{\cM,s} (\rew) = 0$, i.e., from a state $s$ from which no positive reward is reachable, is irrelevant. 
So, we can collapse all these states $s$ with $\mathbb{E}^{\max}_{\cM,s} (\rew) = 0$ (together with the new trap state) to a single trap state that we call $\goal$. 
By these constructions, we obtain a new MDP $\cM^\prime$ in which exactly the same distributions of the total reward can be realized by schedulers as in $\cM$.
As $\cM^\prime$ does not contain any end components anymore and $\goal$ is the only  trap state in $\cM^\prime$, the state $\goal$ is now reached with probability $1$ under any scheduler.
In the light of the described constructions, we work under the following assumption:

\begin{assumption}
\label{ass:1}
W.l.o.g., we assume that all MDPs have a trap state $\goal$, which is reached with probability $1$ under all schedulers. We add this trap state to the signature and hence denote MDPs $\cM$ as tuples
$\cM= (S,\Act, P , \sinit, \rew, \goal)$.
\end{assumption}

All objectives  studied in this paper depend only on the distribution of the random variable $\rew$.
By the following lemma, which is folklore and follows from the  formulation in \cite[Lemma 2]{Piribauer2022TheVS}  (see Appendix \ref{app:prelim}), we can restrict ourselves to reward-based schedulers.
\begin{restatable}{lemma}{lemrewardbased}
\label{lem:reward-based}
Let $\cM = (S,\Act, P , \sinit, \rew, \goal)$ be an MDP satisfying Assumption \ref{ass:1}. 
Then, for any scheduler $\sched$ there is a reward-based scheduler $\tsched$ such that the distribution of the random variable $\rew$ is the same under the probability measures 
$\Pr_{\cM}^{\sched}$ and $\Pr_{\cM}^{\tsched}$.
\end{restatable}

\section{Mean absolute deviation-penalized expectation}
\label{sec:MAD}

As described in the introduction, the VPE suffers from the drawback that optimal schedulers are \ERMin-schedulers, which is an undesirable behavior.
Intuitively, the reason for this behavior in the case of VPE lies in the fact that  the variance grows quadratically with the distance to the expected value.
A natural alternative is  choosing the absolute distance rather than the quadratic distance from the expected value as the measure for the penalty. 
So, we define the \textit{mean absolute deviation} (MAD) of a random variable $X$ as the probability-weighted sum of the distance to the expected value: 
		$
		\MAD(X) \eqdef \mathbb{E}(|X - \mathbb{E}(X)|) $.
		
		We  consider the MAD-penalized expectation (MADPE) of the accumulated weight in an MDP $\cM=(S, \Act, P ,\sinit, \rew,\goal)$ analogously to the VPE:
		We define the MAD of the accumulated reward $\rew$ under scheduler $\sched$ as
		$
		\MAD_{\cM}^{\sched}(\rew) \eqdef \mathbb{E}^{\sched}_{\cM} \left( \left| \rew -\mathbb{E}^{\sched}_{\cM} (\rew) \right| \right)$.
		The MAD-penalized expectation with parameter $\lambda\in \mathbb{R}$ is now 
		 $
		 \MADPE[\lambda]^{\sched}_{\cM}(\rew) \eqdef  \mathbb{E}_{\cM}^{\sched}(\rew) - \lambda \MAD_{\cM}^{\sched}(\rew)$
		  analogously to the VPE. Our goal is to find 
		  \begin{center}
		  $
		   \MADPE[\lambda]^{\max}_{\cM}(\rew) \eqdef  \sup_{\sched}  \MADPE[\lambda]^{\sched}_{\cM}(\rew)
		  $
		  \end{center}
		  as well as an optimal scheduler.
		  In the sequel, we will prove the following results. Proofs omitted in this section can be found in Appendix \ref{app:MAD}.
		  \begin{enumerate}
		  \item
		  In general, randomization is necessary to optimize the MADPE.
		  \item
		  If $\lambda>\frac{1}{2}$, then there is an MDP $\cM$  such that any optimal scheduler for the MADPE is an \ERMin-scheduler.
		  \item
		  If $\lambda\leq\frac{1}{2}$, for any MDP $\cM$, optimal schedulers can be chosen to be reward-based \ERMax-schedulers.
		  \item
		  If $\lambda\leq \frac{1}{2}$, the optimal MADPE can be computed in exponential time.
		  \item 
		  Even for acyclic Markov chains, deciding whether the MADPE exceeds a given threshold $\vartheta$ is PP-hard under polynomial-time Turing reductions.
		  \end{enumerate}

		  \subsection{Randomization and optimality of \ERMin-schedulers}
		  \label{sec:ranopt}

		  \begin{figure}[t]
		  \begin{subfigure}[b]{0.48\textwidth}
\centering
    \resizebox{.9\textwidth}{!}{%
      \begin{tikzpicture}[scale=1,auto,node distance=8mm,>=latex]
        \tikzstyle{round}=[thick,draw=black,circle]

        \node[round, draw=black,minimum size=10mm] (goal) {$\goal$};
        \node[round, below=11mm of goal, minimum size=10mm] (s2) {$s_1$};
        \node[round, left=20mm of s2, minimum size=10mm] (s1) {$s_0$};
        \node[round, right=20mm of s2, minimum size=10mm] (s3) {$s_2$};        
        \node[round, below=11mm of s2, minimum size=10mm] (init) {$\sinit$};
        
         \draw[color=black ,->,very thick] (init)  edge  node [very near start, anchor=center] (m6) {} node [pos=0.7,left=3pt] {$1/4$} node [pos=0.1,left=10pt] {$\alpha\colon +0$} (s1) ;
        \draw[color=black ,->,very thick] (init) edge [bend left=25]  node [very near start, anchor=center] (m5) {} node [pos=0.7,left] {$3/4$} (s2) ;
        \draw[color=black , very thick] (m5.center) edge [bend right=35]  (m6.center);
        
        \draw[color=black ,->,very thick] (init)  edge  node [very near start, anchor=center] (k6) {} node [pos=0.7,right=3pt] {$1/4$}  node [pos=0.1,right=10pt] {$\beta\colon +0$}  (s3) ;
        \draw[color=black ,->,very thick] (init) edge [bend right=25] node [very near start, anchor=center] (k5) {} node [pos=0.7,right] {$3/4$} (s2) ;
        \draw[color=black , very thick] (k5.center) edge [bend left=35] (k6.center);

        \draw[color=black ,->,very thick] (s1) edge [bend left=10]  node [pos=0.5,left=3pt] {$\tau: 0 $} (goal) ;
                \draw[color=black ,->,very thick] (s2) edge  node [pos=0.5,left] {$\tau: +1 $} (goal) ;
            \draw[color=black ,->,very thick] (s3) edge [bend right=10]  node [pos=0.5,right=3pt] {$\tau: +2 $} (goal) ;

      \end{tikzpicture}
    }

  \caption{The MDP $\cM$ used in Example \ref{exam:randomization}.}
  \label{fig:randomization}
  \end{subfigure}
  \begin{subfigure}[b]{0.48\textwidth}
\centering
    \resizebox{.9\textwidth}{!}{%
      \begin{tikzpicture}[scale=1,auto,node distance=8mm,>=latex]
        \tikzstyle{round}=[thick,draw=black,circle]

        \node[round, draw=black,minimum size=10mm] (goal) {$\goal$};
        \node[round, below=11mm of goal, minimum size=10mm] (s2) {$s_1$};
        \node[round, left=20mm of s2, minimum size=10mm] (s1) {$s_{\mathit{dec}}$};
        \node[ right=20mm of s2, minimum size=10mm] (s3) {};        
        \node[round, below=11mm of s2, minimum size=10mm] (init) {$\sinit$};
        
         \draw[color=black ,->,very thick] (s2)  edge  node [very near start, anchor=center] (m6) {} node [pos=0.4,below=3pt] {$1/2$} node [pos=0.45,above=5pt] {$\tau\colon +1$} (s1) ;
        \draw[color=black ,->,very thick] (s2) edge [loop above]  node [very near start, anchor=center] (m5) {} node [pos=0.4,left] {$1/2$} (s2) ;
        \draw[color=black , very thick] (m5.center) edge [bend right=35]  (m6.center);
        
        \draw[color=black ,->,very thick] (init)  edge  [bend right=60] node [very near start, anchor=center] (k6) {} node [pos=0.3,right=3pt] {$1-p$}  node [pos=0.1,right=10pt] {$\tau\colon +0$}  (goal) ;
        \draw[color=black ,->,very thick] (init) edge node [very near start, anchor=center] (k5) {} node [pos=0.7,right] {$p$} (s2) ;
        \draw[color=black , very thick] (k5.center) edge [bend left=35] (k6.center);

        \draw[color=black ,->,very thick] (s1) edge  node [pos=0.5,left=3pt] {$\beta: 0 $} (goal) ;
                \draw[color=black ,->,very thick] (s1) edge [bend left=40]  node [pos=0.5,left=3pt] {$\alpha: +1 $} (goal) ;

      \end{tikzpicture}
    }

  \caption{The MDP $\cM$ used in Example \ref{exam:minimizing}.}
  \label{fig:minimizing}
\end{subfigure}
\caption{Two example MDPs.}
\vspace{-12pt}
\end{figure}

		  We work with MDPs $\cM=(S,\Act,P,\sinit,\rew,\goal)$ satisfying Assumption \ref{ass:1}.
First, we show that randomization is necessary for the optimization of the MADPE in the following example.

\begin{example}
\label{exam:randomization}
Consider the MDP $\cM$  in Figure \ref{fig:randomization}. 
We consider the schedulers $\sched_\alpha$ choosing $\alpha$ in $\sinit$, $\sched_\beta$ choosing $\beta$, and $\sched_{1/2}$ choosing $\alpha$ and $\beta$ with probability $1/2$ each
and obtain:
$
\mathbb{E}^{\sched_{\alpha}}_{\cM}(\rew) = 3/4$,  $\mathbb{E}^{\sched_{1/2}}_{\cM}(\rew) = 1$, and $  \mathbb{E}^{\sched_{\beta}}_{\cM}(\rew) = 5/4$.
The MADs are
$
\MAD^{\sched_{\alpha}}_{\cM}(\rew) =  3/8$, $\MAD^{\sched_{1/2}}_{\cM}(\rew)= 1/4\cdot 1 = 1/4$, and $ \MAD^{\sched_{\beta}}_{\cM} (\rew)= 3/8$.
Clearly, the MADPE under $\sched_\beta$ is better than under $\sched_{\alpha}$ for any  $\lambda>0$. For the MADPE of $\sched_{1/2}$ and $\sched_\beta$ with $\lambda=4$, we obtain
\begin{center}
$
\MADPE[\lambda]^{\sched_{1/2}}_{\cM}(\rew) = 1- \frac{1}{4}\lambda  = 0, \quad \MADPE[\lambda]^{\sched_\beta}_{\cM}(\rew)  = \frac{5}{4} - \frac{3}{8}\lambda = -\frac{1}{4} .
$
\end{center}
So, the randomized scheduler $\sched_{1/2}$ is better than the deterministic schedulers $\sched_{\alpha}$ and $\sched_{\beta}$.
In Figure \ref{fig:MADVAR}, we  depict the MAD in comparison to the expected value of any randomized scheduler for $\cM$. The kink in the graph at expected value $1$ can be explained by the fact that
the MAD contains a summand for $|1-\mathbb{E}^{\sched}_{\cM}(\rew)|$.
The dotted blue line consists of all points in the $\MAD$-$\mathbb{E}$-plane with the same MADPE as the scheduler $\sched_{1/2}$ illustrating that this scheduler is in fact optimal as the MADPE increases in the direction of the arrow.
For comparison, we also depict the variances of randomized schedulers over the expectation. Clearly, for any $\lambda$ the deterministic scheduler choosing $\beta$ will always be VPE-optimal.
\end{example}

		  In the next example, we will illustrate that the MADPE fails to guarantee in general that optimal schedulers are eventually reward-maximizing.

		  \begin{example}
		  \label{exam:minimizing}

		Consider the MDP $\cM$ depicted in Figure \ref{fig:minimizing} for $p\in (0,1/3]$. 
		Always choosing $\alpha$ in state $s_{\mathit{dec}}$ maximizes the expected value. Under this scheduler, the expected value is $3p \leq 1$ as moving from state $s_1$ to state $s_{\mathit{dec}}$ takes two steps in expectation.
		So, under any scheduler, the expected value lies between $0$ and $1$. So, all paths leading via $s_1$ yield a reward above the expected value, while only the path going directly to $\goal$ from $\sinit$ yields a reward below the expected value. For the MAD under a scheduler $\sched$, we obtain $\MAD^{\sched}_{\cM}(\rew) = 2 \cdot (1-p) \cdot  \mathbb{E}^{\sched}_{\cM}(rew)$ (see Appendix \ref{app:MAD} for the calculations).

		For a given $\lambda>\frac{1}{2}$, we can choose  $p\in (0,1/3]$ such that $\lambda> \frac{1}{2(1-p)}$ and hence $\lambda \cdot 2 \cdot (1-p)>1$. Now, under any scheduler $\sched$, the MADPE for parameter $\lambda$ is
		\[
		\MADPE[\lambda]^{\sched}_{\cM}(\rew) = \mathbb{E}^{\sched}_{\cM}(rew) - \lambda \cdot 2 \cdot (1-p) \cdot  \mathbb{E}^{\sched}_{\cM}(rew) = (1- \lambda \cdot 2 \cdot (1-p) ) \mathbb{E}^{\sched}_{\cM}(rew).
		\]
		As $1- \lambda \cdot 2 \cdot (1-p) <0$, a scheduler maximizing the MADPE  has to minimize the expected value of $\rew$. In $\cM_p$, this means  always choosing $\beta$.
		So,  for any $\lambda>\frac{1}{2}$, there is an MDP in which  optimal schedulers have to minimize the future expected rewards no matter how large the accumulated reward already is.

		\end{example}

		    \begin{figure}[t]
		    \begin{subfigure}[b]{.40\textwidth}
		     \resizebox{\textwidth}{!}{
    \begin{tikzpicture}[yscale=1,xscale=1]
      \begin{axis}[
       width=7cm,height=4.5cm,
          axis lines = middle,
          xlabel={$\mathbb{E}$},
          ylabel={$\MAD$},
          ymin=0, ymax=.4,
          xmin=.5, xmax=1.5
                ]
        \addplot [name path=A,domain=0.75:1,
          samples=300,
          color=black,thick]
        {(3+8*(x-.75)-16*(x-.75)*(x-.75))/16 + 3/16 -3/4*(x-.75)};
         \addplot [name path=B,domain=1:1.25,
          samples=300,
          color=black,thick]
        {(3+8*(x-.75)-16*(x-.75)*(x-.75))/16 - 3/16 +3/4*(x-.75)};

        \node[label={180:{$\alpha$}},circle,fill,inner sep=2pt] at (axis cs:.75,.375) {};
        \node[label={0:{$\beta$}},circle,fill,inner sep=2pt] at (axis cs:1.25,0.375) {};

        \draw[dashed,thick,color=blue] (axis cs:.5,.125) -- (axis cs:1.5,.375);
      \draw[thick,color=blue,->] (axis cs:1.2,.3) -- (axis cs:1.25,.19);
      \node[label={0:{\textcolor{blue}{$\mathbb{E}-\lambda \cdot \MAD$}}}] at (axis cs:1.1,.17) {};
      \end{axis}
    \end{tikzpicture}
    }
  \end{subfigure}
  \hspace{24pt}
    \begin{subfigure}[b]{.40\textwidth}
		     \resizebox{\textwidth}{!}{
    \begin{tikzpicture}[yscale=1,xscale=1]
      \begin{axis}[
       width=7cm,height=4.5cm,
          axis lines = middle,
          xlabel={$\mathbb{E}$},
          ylabel={$\Var$},
          ymin=0, ymax=.4,
          xmin=.5, xmax=1.5
                ]
        \addplot [name path=A,domain=0.75:1.25,
          samples=300,
          color=black,thick]
        {0.25-(1-x)*(1-x)};

        \node[label={270:{$\alpha$}},circle,fill,inner sep=2pt] at (axis cs:.75,0.1875) {};
        \node[label={270:{$\beta$}},circle,fill,inner sep=2pt] at (axis cs:1.25,0.1875) {};

      \node[label={0:{\textcolor{blue}{$\mu-1\cdot \sigma^2$}}}] at (axis cs:3.7,1.3) {};
      \end{axis}
    \end{tikzpicture}
    }
  \end{subfigure}
  \caption{Plot of MAD and variance over the expected value for schedulers obtained by choosing $\alpha$ with probability $p\in[0,1]$ in the MDP $\cM$ depicted in Figure \ref{fig:randomization}.}
\vspace{-12pt}
\label{fig:MADVAR}
\end{figure}
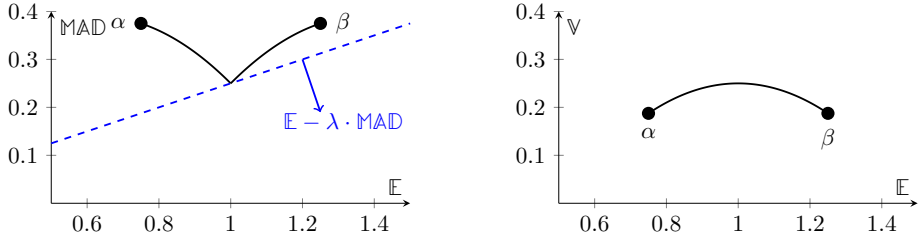

		\subsection{Sufficiently small parameters $\lambda$}
		
	As we have seen, the MADPE as an objective does not in general guarantee that optimal schedulers are \ERMax-schedulers. In this section, we now show that this desirable property is guaranteed if the risk-aversion parameter $\lambda$ is at most $\frac{1}{2}$.
	
	By Lemma \ref{lem:reward-based}, we already know that we can restrict ourselves to reward-based schedulers when optimizing the MADPE. 
	For two reward-based schedulers $\sched$ and $\tsched$ and a natural number $k$, we define the  reward-based scheduler $\sched \uparrow_k \tsched$ on state-reward-pairs $(s,w)\in S\times \mathbb{N}$ by
		$
		(\sched \uparrow_k \tsched) (s,w) = \begin{cases}
		\sched(s,w) & \text{ if $w<k$,} \\
		\tsched (s,w) & \text{ if $w\geq k$}
		\end{cases}
		$	
		where we view $\sched$ and $\tsched$ as  functions from $S\times \mathbb{N}$ to  distributions over actions.
		
		For risk-aversion parameters $\lambda$ of at most $1/2$, the following theorem implies that optimal schedulers for the MADPE can be chosen to be \ERMax-schedulers.

		\begin{restatable}{theorem}{thmMADoptimalsched}
		\label{thm:MAD_optimal_sched}
		Let $\cM=(S,\Act, P, \sinit, \rew, \goal)$ be an MDP satisfying Assumption \ref{ass:1} and let $\lambda \in (0,\frac{1}{2}]$ be  a parameter for the MADPE.
		Further, let $\tsched$ be a memoryless deterministic scheduler with $\mathbb{E}^{\tsched}_{\cM,s}(\rew)=\mathbb{E}^{\max}_{\cM,s}(\rew)$. Let $k=\lceil \mathbb{E}^{\max}_{\cM} (\rew) \rceil$.
		Then, for any reward-based scheduler $\sched$, we have
		$
		\MADPE[\lambda]_{\cM}^{\sched}(\rew) \leq \MADPE[\lambda]_{\cM}^{\sched \uparrow_k \tsched }(\rew)$.
				\end{restatable}
				The theorem is shown by expressing the MADPE using conditional expectations under the condition that the reward exceeds the bound~$k$. Note that the theorem  implies that it does not matter which expectation optimal scheduler $\tsched$ is chosen after a reward of at least $\mathbb{E}^{\max}_{\cM} (\rew) $ has been accumulated.

		\subsection{Computing the maximal MADPE}

		Theorem \ref{thm:MAD_optimal_sched} tells us that the value $\MADPE[\lambda]^{\max}_{\cM}$ in an MDP $\cM$ for $\lambda\in (0,1/2]$ is the supremum of $\MADPE[\lambda]^{\sched}_{\cM}$
		over all reward-based schedulers $\sched$ that behave according to a fixed memoryless deterministic scheduler $\tsched$ maximizing the expected reward as soon as a reward of more than $\mathbb{E}^{\max}_{\cM}(\rew)$ 
		has been accumulated. Let us denote the set of such schedulers by $\Sched_{\cM}^\tsched$.  

		The result  shares some similarity with the results in \cite{tacas2017} on the computation of maximal conditional expected rewards under the condition that a set of target states is reached. In both cases, a reward-based scheduler that has to keep track of the accumulated reward up to some bound $B$ has to be computed. The bound $B$, however, is obtained quite differently. Here, the maximal expected accumulated reward can be used as this bound.  The bound in \cite{tacas2017} is in general much larger (although also exponential). Similar reward-based schedulers are also necessary for the model-checking of temporal formulas with certain reward operators \cite{DBLP:journals/fmsd/ChenFKPS13} and for the optimization of the variance-penalized expectation \cite{Piribauer2022TheVS}.

		We are now in the position to provide a model transformation such that afterwards we can restrict ourselves to memoryless schedulers.
		Given the MDP $\cM=(S,\Act, P, \sinit, \rew, \goal)$, let $k=\lceil \mathbb{E}^{\max}_{\cM}(\rew) \rceil$ and let $\ell$ be the largest reward of a state-weight pair in $\cM$.
		We now define the MDP $\cN = (S^\prime, \Act^\prime, P^\prime, \sinit^\prime, \rew^\prime, \goal^\prime)$.
		
		The state space $S^\prime = S\times \{0,\dots, k+\ell-1\}\cup \{\goal^\prime\}$ and represents states together with the reward that has been accumulated so far, as well as a new trap state $\goal^\prime$. 
		The initial state is $\sinit^\prime = (\sinit, 0)$.
		The set of actions is 
		extended by one new action $\tau$.
		The transition probability function $P^\prime$ for $(s,w)\in S\times \{0,\dots, k+\ell-1\}$ and $\alpha \in \Act$ is given by
		$
		P^{\prime}((s,w),\alpha,(t,v)) = 
		P(s,\alpha,t)$  if $w\leq k-1$ and $v=w+\rew(s,\alpha)$, and is set to $0$ otherwise.
		So, in all states in $S\times \{k,\dots, k+\ell-1\}$ and in $\{\goal \} \times \{0,\dots, k-1\}$ none of the actions in $\Act$ are enabled. Instead in 
	 these states the new action $\tau$ is enabled and leads to 
		the  trap state $\goal^\prime$ with probability $1$.
		The reward function is $0$ on all state-action pairs containing an action from $\Act$. Only the new action $\tau$ gets assigned a reward by
		\begin{align*}
		&\rew^\prime((\goal,w))=w \,\, && \text{ for all $w\in \{0,\dots,k+\ell-1\}$} \quad \text{ and }
		\\
		&\rew^\prime ((s,w)) =  w+ \mathbb{E}^{\max}_{\cM,s}(\rew) \,\, && \text{ for $s\in S\setminus \{\goal\}$ and $w\in \{k,\dots, k+\ell-1\}$.}
		\end{align*}
		So, in $\cN$, rewards are only received in the very last step when entering the trap state $\goal^\prime$.
		
		Now, a scheduler $\sched \in \Sched_{\cM}^\tsched$ for $\cM$ can be seen as a memoryless scheduler for $\cN$ and vice versa:
		The scheduler $\sched$ makes decision for all state-reward pairs $(s,w)$ with $s\not=\goal$ and $w <\mathbb{E}^{\max}_{\cM}(\rew)$. 
		For higher values of accumulated reward, it switches to the behavior of the memoryless scheduler $\tsched$. 
		A memoryless scheduler for $\cN$ has to choose a probability distribution over $\Act$ on the same pairs $(s,w)$. For higher values of $w$ or for pairs $(\goal,w)$, 
		 only action $\tau$ is enabled in $\cN$.
		So, with a slight abuse of notation, we interpret schedulers in $\Sched_{\cM}^\tsched$ for $\cM$ also as memoryless schedulers for $\cN$ and vice versa.
		\begin{remark}
		As  reward-based schedulers are sufficient to maximize the MADPE and in $\cN$ rewards are only received in the last step, we can conclude that 
		memoryless schedulers are sufficient to maximize the MADPE in $\cN$.
		\end{remark}
		
		\begin{restatable}{lemma}{lemcorrectnessconstruction}
		Given $\cM$ and $\cN$ as above, a scheduler $\sched \in \Sched_{\cM}^\tsched$ and $\lambda\in (0,1/2]$, we have
		$
		\MADPE[\lambda]^{\sched}_{\cM}(\rew) = \MADPE[\lambda]^{\sched}_{\cN}(\rew^\prime)$.
		\end{restatable}

		\noindent We  utilize the MDP $\cN$  to compute the maximal MADPE  via a quadratic program: 
		\begin{theorem}
		Let $\cM$ be an MDP with non-negative rewards and $\lambda\in (0,1/2]$. Then,  $\MADPE[\lambda]^{\max}_{\cM}$ is the optimal solution 
		to a linearly-constrained quadratic program that can be constructed from $M$ and $\lambda$ in exponential time.
		\end{theorem}
		
		Note that the MDP $\cN$ can be constructed in exponential time from $\cM$ as the numerical value of the maximal expected value $\mathbb{E}^{\max}_{\cM}(\rew)$ is at most exponentially large in the size of $\cM$.
		So, it is sufficient to construct a quadratic program from $\cN$ in polynomial time.
		In the sequel, we provide the construction of the quadratic program and prove its correctness.

		We start by providing linear constraints that specify the possible combinations of expected frequencies of state-action-pairs under some scheduler.
		We  use variables $x_{s,w,\alpha}$ for all $s\in S$, $w\in  \{0,1,\dots, k+\ell - 1\}$, and $\alpha\in \Act^\prime((s,w))$. 
		For  these variables, we require
		\begin{align}
		\label{eq:c1}
		x_{s,w,\alpha} & \geq 0, \quad \text{ and}\\
		\label{eq:c2}
		\sum_{\alpha\in \Act(s)} x_{s,w,\alpha} & = \sum_{t\in S, \beta\in \Act(t)} x_{t,w-\rew(t,\beta) , \beta} \cdot P(t,\beta,s) + \mathbb{1}_{(s,w)= (\sinit,0)}
		\end{align}
		where  $\mathbb{1}_{(s,w)= (\sinit,0)}=1$ iff $s=\sinit$ and $w=0$, and $\mathbb{1}_{(s,w)= (\sinit,0)}=0$ otherwise.
		In any solution to these two constraints, the variables $x_{s,w,\alpha}$ represent the expected frequency with which action $\alpha$ is chosen in state $(s,w)$ under some scheduler.
		This is made precise below.

		Rewards are only accumulated on the final transitions from a state $(s,w)$ to $\goal^\prime$ via action $\tau$ for $s=\goal$ or $w\geq k$.
		As these transitions lead to the absorbing state with probability $1$, the expected frequency with which the action $\tau$ is chosen is  the probability with which the respective transition is taken.
		 So, we can encode the expected value in  an auxiliary variable  $e$ defined via the constraint 
		\begin{align}
		\label{eq:c4}
		e= \sum_{w=0}^{k-1} x_{\goal,w,\tau}\cdot w + \sum_{w=k}^{k+\ell -1} \sum_{s\in S} x_{s,w,\tau}\cdot (w +\mathbb{E}^{\max}_{\cM,s}(\rew)).
		\end{align}

		\begin{restatable}{lemma}{lemfrequency}
		\label{lem:frequency}
		For any solution vector to constraints (\ref{eq:c1}) -- (\ref{eq:c4}), there is a scheduler $\sched$ for $\cN$ such that
		$
		\Pr^{\sched}_{\cN} (\Diamond (s,w)) = x_{s,w,\tau} $ for all $(s,w)$ with $s= \goal$ or $w\geq k$
		and such that $ \mathbb{E}^{\sched}_{\cM}(\rew) = e$;
		and vice versa.
		\end{restatable}

		Now, we can use these auxiliary variables to encode the MADPE as an objective function:
		
		\begin{align}
		\label{eq:obj_abs}
		 \text{maximize}\quad & e - \lambda \left( \sum_{w=0}^{k-1}  x_{\goal,w,
		tau}\cdot |w - e|   +  \sum_{w=k}^{k+\ell -1} \sum_{s\in S} x_{s,w,\tau}\cdot \left| w +\mathbb{E}^{\max}_{\cM,s}(\rew) - e\right| \right) 
		\end{align}
		This function still contains the absolute value operator. However, all absolute value terms occur with a negative sign. Therefore, we can use further variables $g_i$ for $i\in \{0,\dots, k-1\}$ and
		$h_{s,w}$ for $(s,w)\in S\times \{k,\dots, k+\ell -1\}$ to capture the absolute value. The following constraints state that these variables are at least as big as the respective absolute value terms.
		For $w\in \{0,\dots, k-1\}$, we require
		\begin{align}
		\label{eq:c5}
		g_w \geq w - e \quad \text{ and }\quad -g_w \leq w - e.
		\end{align}
		For $(s,w)\in S\times \{k,\dots, k+\ell -1\}$, we require 
		\begin{align}
		\label{eq:c6}
		h_{s,w} \geq w +\mathbb{E}^{\max}_{\cM,s}(\rew) - e \quad \text{ and }\quad -h_{s,w} \leq w +\mathbb{E}^{\max}_{\cM,s}(\rew) - e.
		\end{align}
		The new objective function can now be written as
		\begin{align}
		\label{eq:objective}
		\text{maximize }\quad & e - \lambda \left( \sum_{w=0}^{k-1}  x_{\goal,w,\tau}\cdot g_w   +  \sum_{w=k}^{k+\ell -1} \sum_{s\in S} x_{s,w,\tau}\cdot  h_{s,w} \right).
		\end{align}
		
		\begin{restatable}{theorem}{thmquadraticprogram}
				The optimal solution to (\ref{eq:objective}) under constraints (\ref{eq:c1}) - (\ref{eq:c4}), (\ref{eq:c5}), and (\ref{eq:c6}) is the maximal MADPE $\MADPE[\lambda]^{\max}_{\cN}$.
		\end{restatable}
		
		\begin{proof}
		As all variables are non-negative, the variables $g_w$ with $0\leq w \leq k-1$   and $h_{s,w}$ with $w\geq k$  in the objective function  (\ref{eq:objective}) occur under a negative sign. 
		To maximize the objective function, these variables hence have to be set to the minimal possible values given the value of the variable $e$. 
		By constraints (\ref{eq:c5}) and (\ref{eq:c6}), these minimal possible values are the values $|w-e|$ and $|w +\mathbb{E}^{\max}_{\cM,s}(\rew) - e|$, respectively.
		So, the optimal value of this quadratic objective function is the same as of the objective function (\ref{eq:obj_abs}), which directly encodes the MADPE. 
		\end{proof}

\subsection{Computational hardness of the MADPE}
		
		The complexity class PP \cite{Gill1977} is characterized as the class of languages $\mathcal{L}$ that have a probabilistic polynomial-time bounded Turing machine $M_{\mathcal{L}}$ such that  $\tau \in \mathcal{L}$ if and only if $M_{\mathcal{L}}$ accepts $\tau$ with probability at least $1/2$ for all words $\tau$.
		We will show PP-hardness under polynomial-time Turing reductions. So, for the reduction, we allow querying an oracle for the problem we reduce to.
A polynomial time algorithm for a problem that is PP-hard under polynomial Turing reductions would imply that the polynomial hierarchy collapses~\cite{Toda1991}.

		\begin{restatable}{theorem}{thmPPMAD}
		\label{thm:PP-hard_MAD}
		Deciding for an acyclic Markov chain $\cM$ and a threshold $\vartheta\in \mathbb{Q}$ whether $\MAD_{\cM}(\rew) \geq \vartheta$ is PP-hard under polynomial-time Turing reductions.
		\end{restatable}
		
		\begin{proof}[Proof sketch]
		We reduce from the following  problem that is shown to be PP-hard in \cite{HaaseK15}: 
		Given an acyclic Markov chain $\cM=(S,P,\sinit, \rew)$, and a natural number $t$, decide whether $\Pr_{\cM}(\rew>t) \geq 1/2$.	
		We first show that the exact value $\MAD_{\cM}(\rew)$ can be computed in acyclic Markov chains via a binary search using polynomially many calls to an oracle for the threshold problem.
		Then, we prove that $\Pr_{\cM}(\rew>t)$ can be computed by comparing the MAD in two variations of $\cM$ that ensure that the expected value of $\rew$ in these variations is $t$ and $t+1/2$, respectively.  
			\end{proof}

		\begin{corollary}
		Deciding for an acyclic Markov chain $\cM$,  $\lambda\in \mathbb{Q}_+$ and  $\vartheta\in \mathbb{Q}$ if $\MADPE[\lambda]_{\cM}(\rew)\geq \vartheta$ is PP-hard under polynomial-time Turing reductions.
		\end{corollary}

\section{Semi-deviation measure-penalized expectation}
\label{sec:semivariance}

To overcome the restrictions on the parameter $\lambda$ for the MADPE or to overcome the undesirable behavior observed for the VPE, one might be tempted to consider the semi-MAD (SMAD) or the semi-variance as a deviation measure that only considers outcomes below the expected value as a measure for the penalty.

\paragraph*{Semi-MAD-penalized expectation.}
		\label{sec:semiMAD}
		 We define
		$
		\SMAD(X) = \mathbb{E}(\max(0, \mathbb{E}(X) - X))
		$
		for a random variable $X$. So, all outcomes above the expected value do not contribute to the SMAD.
		However, the SMAD is always half the MAD, i.e., $\SMAD(X)=\MAD(X)/2$, as one can easily compute (see Appendix \ref{app:SMAD}).
				So, using the SMAD as a penalty term is the same as using the MAD besides a rescaling of the penalty factor $\lambda$ by a factor of $2$.

				\paragraph*{Semi-variance-penalized expectation (SVPE).}

		We now define  the {semi-variance},  to only treat outliers below the expected value with a quadratic penalty. However we will see that SVPE-optimal schedulers might still have to be \ERMin-schedulers.
		We define the semi-variance by ignoring outliers above the expected value as follows
$			\SV(X) \coloneqq \mathbb{E}\big(\big(\min\big(X - \mathbb{E}(X), 0\big)\big)^2\big)$.
		Applied to the accumulated reward in an MDP $\cM=(S,\Act, P, \sinit, \rew,\goal)$, we define
		$
		\SV^{\sched}_{\cM}(\rew) \coloneqq \mathbb{E}^\sched_{\cM}\big(\big(\min\big(\rew - \mathbb{E}^\sched_{\cM}(\rew), 0\big)\big)^2\big)
		$ for schedulers $\sched$. Using this as a penalty, we obtain the SVPE  for a parameter $\lambda$
		\[
		\SVPE[\lambda]^\sched_{\cM}(\rew) = \mathbb{E}^{\sched}_{\cM}(\rew)-\lambda \cdot \SV^{\sched}_{\cM}(\rew)
		\]
		and define the optimal value $\SVPE[\lambda]^{\max}_{\cM}(\rew)$ as usual.
		Besides the possible necessity of \ERMin-schedulers, we will  see that randomization is necessary to optimize the SVPE in contrast to the VPE, for which  optimal deterministic (finite-memory) schedulers exist \cite{Piribauer2022TheVS}.

		  \begin{figure}[t]
		  \begin{subfigure}[b]{0.48\textwidth}
		  \centering
    \resizebox{.9\textwidth}{!}{%
      \begin{tikzpicture}[scale=1,auto,node distance=8mm,>=latex]
        \tikzstyle{round}=[thick,draw=black,circle]

        \node[round, draw=black,minimum size=10mm] (goal) {$\goal$};
        \node[round, below=11mm of goal, minimum size=10mm] (s2) {$s_1$};
        \node[round, left=20mm of s2, minimum size=10mm] (s1) {$s_{\mathit{dec}}$};
        \node[ right=20mm of s2, minimum size=10mm] (s3) {};        
        \node[round, below=11mm of s2, minimum size=10mm] (init) {$\sinit$};
        
         \draw[color=black ,->,very thick] (s2)  edge  node [very near start, anchor=center] (m6) {} node [pos=0.4,below=3pt] {$1/2$} node [pos=0.45,above=5pt] {$\tau\colon +k$} (s1) ;
        \draw[color=black ,->,very thick] (s2) edge [loop above]  node [very near start, anchor=center] (m5) {} node [pos=0.4,left] {$1/2$} (s2) ;
        \draw[color=black , very thick] (m5.center) edge [bend right=35]  (m6.center);
        
        \draw[color=black ,->,very thick] (init)  edge  [bend right=60] node [very near start, anchor=center] (k6) {} node [pos=0.3,right=3pt] {$1/2$}  node [pos=0.1,right=10pt] {$\tau\colon +0$}  (goal) ;
        \draw[color=black ,->,very thick] (init) edge node [very near start, anchor=center] (k5) {} node [pos=0.7,right] {$1/2$} (s2) ;
        \draw[color=black , very thick] (k5.center) edge [bend left=35] (k6.center);

        \draw[color=black ,->,very thick] (s1) edge  node [pos=0.5,left=3pt] {$\beta: 0 $} (goal) ;
                \draw[color=black ,->,very thick] (s1) edge [bend left=40]  node [pos=0.5,left=3pt] {$\alpha: +1 $} (goal) ;

      \end{tikzpicture}
    }

  \caption{The MDP $\cM$ used in Example \ref{exam:minimization_SVPE}.}
   \label{fig:minimizing_SPE}
  \end{subfigure}
  \begin{subfigure}[b]{0.48\textwidth}

\centering
    \resizebox{.9\textwidth}{!}{%
      \begin{tikzpicture}[scale=1,auto,node distance=8mm,>=latex]
        \tikzstyle{round}=[thick,draw=black,circle]

        \node[round, draw=black,minimum size=10mm] (goal) {$\goal$};
        \node[round, below=11mm of goal, minimum size=10mm] (s2) {$s_1$};
        \node[round, left=20mm of s2, minimum size=10mm] (s1) {$s_0$};
        \node[round, right=20mm of s2, minimum size=10mm] (s3) {$s_2$};        
        \node[round, below=11mm of s2, minimum size=10mm] (init) {$\sinit$};
        
         \draw[color=black ,->,very thick] (init)  edge  node [very near start, anchor=center] (m6) {} node [pos=0.7,left=3pt] {$1/2$} node [pos=0.1,left=10pt] {$\alpha\colon +0$} (s1) ;
        \draw[color=black ,->,very thick] (init) edge  node [very near start, anchor=center] (m5) {} node [pos=0.7,left] {$1/2$} (s2) ;
        \draw[color=black , very thick] (m5.center) edge [bend right=35]  (m6.center);
        
        \draw[color=black ,->,very thick] (init)  edge    node [pos=0.1,right=10pt] {$\beta\colon +0$}  (s3) ;

        \draw[color=black ,->,very thick] (s1) edge [bend left=10]  node [pos=0.5,left=3pt] {$\tau: 0 $} (goal) ;
                \draw[color=black ,->,very thick] (s2) edge  node [pos=0.5,left] {$\tau: +100 $} (goal) ;
            \draw[color=black ,->,very thick] (s3) edge [bend right=10]  node [pos=0.5,right=3pt] {$\tau: +40 $} (goal) ;

      \end{tikzpicture}
    }

  \caption{The MDP $\cM$ used in Example \ref{exam:randomization_SPE}.}
 \label{fig:random_SPE}
\end{subfigure}
\caption{Two example MDPs for phenomena of the SVPE.}
\end{figure}
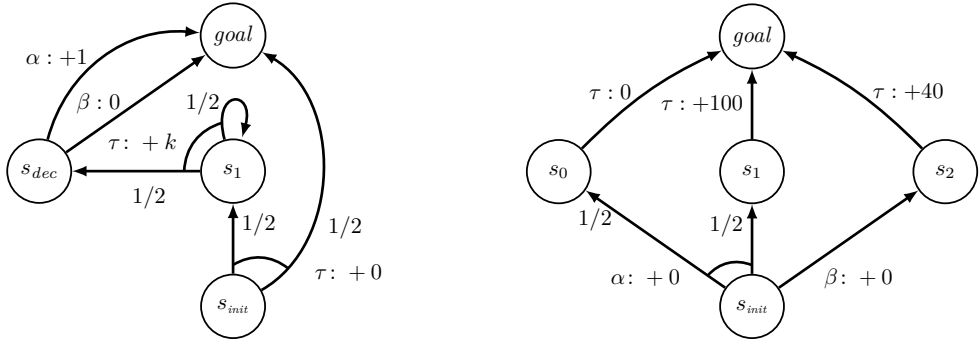

\begin{example}[\ERMin-schedulers]
\label{exam:minimization_SVPE}
Let $\lambda>0$ be a parameter for the SVPE.
Consider the MDP $\cM$ depicted in Figure \ref{fig:minimizing_SPE} where the weight  $k$ is some natural number $k>1/\lambda$.
First, observe that under any scheduler $\sched$, we have $k\leq \mathbb{E}^{\sched}_{\cM}(\rew) \leq k+1/2$.
Now, let  $\ell\geq 2$ be a natural number and let $\sched_p$ be a family of schedulers for $p\in[0,1]$ that behaves exactly the same on all paths except for the path that reaches $s_{\mathit{dec}}$ with
accumulated reward exactly $\ell \cdot k$. In this state, $\sched_p$ chooses $\alpha$ with probability $p$ and $\beta$ with probability $1-p$.

We now want to compare the SVPE of $\sched_p$ for $p>0$ to the SVPE of $\sched_0$. So, let $\lambda>0$ be given.
First, we define $E\coloneqq \mathbb{E}^{\sched_0}_{\cM} (\rew)$ and observe
$
\mathbb{E}^{\sched_p}_{\cM} (\rew) = E + \frac{p}{2^{\ell+1}}
$
as the path on which $\sched_p$ and $\sched_0$ differ has probability $\frac{1}{2^{\ell+1}}$.
Furthermore, both schedulers differ only on a path with a reward higher than the maximal possible expected accumulated reward, which is $k+1/2$. This means that the semivariance under $\sched_p$ will be larger as under $\sched_0$.
Note that exactly the outcomes with reward at most $k$ contribute to the semivariance and these outcomes have exactly the same probability under $\sched_p$ and $\sched_0$. However, the expected value under $\sched_p$ is higher.
We estimate
$
\SV^{\sched_p}_{\cM}(\rew) - \SV^{\sched_0}_{\cM}(\rew) \geq \frac{1}{2} (E + \frac{p}{2^{\ell+1}})^2 - \frac{1}{2}E^2
$
by only considering the increase in the squared distance from the mean for the outcome $0$ that occurs with probability $1/2$ under both schedulers. So, we can conclude
$\SV^{\sched_p}_{\cM}(\rew) - \SV^{\sched_0}_{\cM}(\rew) \geq \frac{1}{2} (\frac{2Ep}{2^{\ell+1}} +( \frac{p}{2^{\ell+1}})^2) \geq \frac{Ep}{2^{\ell+1}}$. 
For the SVPE, this implies
$
\SVPE[\lambda]^{\sched_p}_{\cM}(\rew) - \SVPE[\lambda]^{\sched_0}_{\cM}(\rew) \leq  \frac{p}{2^{\ell+1}} - \lambda \frac{Ep}{2^{\ell+1}}.
$
As $E\geq k$ and  $\lambda > 1/k$, the SVPE under scheduler $\sched_0$ is higher than under $\sched_p$.
Note that $\ell\geq 2$ was chosen arbitrarily. So, this argument shows that any scheduler can be improved by always scheduling $\beta$ in $s_{\mathit{dec}}$ as soon as the accumulated reward is at least $2k$.
\end{example}

For each $\lambda>0$, we have provided an MDP in which  optimal schedulers are necessarily \ERMin-schedulers.
This is exactly the undesirable behavior as for the VPE  we aim to overcome.
So, the SVPE is not a suitable alternative.

\begin{example}
\label{exam:randomization_SPE}
To conclude, we  show that randomization is necessary to maximize the SVPE.
Consider the MDP $\cM$ depicted in Figure \ref{fig:random_SPE}.
Let $\sched_p$ be the scheduler that chooses action $\alpha$ with probability $p$. Further, let $\lambda = \frac{1}{100}$.
We compute
$
\mathbb{E}^{\sched_p}_{\cM}(\rew) = 40+10p$.
Under $\sched_p$,  reward $40$ is accumulated with probability $1-p$ and  reward $0$ with probability $p/2$. So, we obtain
$
\SV^{\sched_p}_{\cM} (\rew) = (1-p)\cdot (10p)^2 + \frac{p}{2} \cdot (40+10p)^2 = 800 p + 500 p^2 - 50 p^3$.
Finally, we compute
$\mathbb{E}^{\sched_p}_{\cM}(\rew) - \lambda \SV^{\sched_p}_{\cM} (\rew)  = 40 +2p -5p^2 + \frac{1}{2} p^3$.
We determine the unique maximum of this expression on the interval $[0,1]$ 
at the zero of its derivative, which lies at $p \approx 0.206$.
So, randomization is necessary in order to maximize the SVPE in this MDP.

To conclude, let us  compute the variance to illustrate that randomization is not increasing the  VPE.
We obtain
$\Var_{\cM}^{\sched_p}(\rew) = \SV^{\sched_p}_{\cM} (\rew) + \frac{p}{2} (60-10p)^2 
= 2600p -100 p^2$. For the VPE for an arbitrary parameter $\lambda>0$, this results in
$
\mathbb{E}^{\sched_p}_{\cM}(\rew) - \lambda \Var^{\sched_p}_{\cM} (\rew)  = 40+10p -2600\lambda p +100 \lambda p^2.
$
Due to the positive coefficient in front of $p^2$ this is a parabola opened upwards. So, for any $\lambda$, one of the deterministic schedulers with $p=0$ or $p=1$ is optimal.

\end{example}

\section{Threshold-based penalty}
\label{sec:threshold}
	
	The MADPE penalizes outcomes below the expected value of the accumulated reward. The computation of the optimal MADPE via a quadratic program of exponential size, however, might not be feasible on large models.
	A conceptually simpler alternative, for which we will be able to provide a pseudo-polynomial optimization algorithm, is to externally fix a threshold $t$ and to penalize outcomes below this threshold $t$. 
	To this end, we define a threshold-based penalty  function $\TBP^\lambda_t\colon \mathbb{R}\to \mathbb{R}$ for parameters $\lambda, t >0$ by
	$
	\TBP^\lambda_t(x) = x - \lambda\cdot \max(t-x,0)$.
	
	This function returns $x$ if $x$ is at least $t$ and otherwise  penalizes the deviation below the value $t$ linearly with the penalty factor $\lambda$.
	In an MDP $\cM$, our goal is now to maximize -- by choosing a scheduler $\sched$ -- the threshold-based-penalized expectation (TBPE) 
	\[
	\mathbb{E}^{\sched}_{\cM}(\TBP^{\lambda}_t (\rew)) = \mathbb{E}^{\sched}_{\cM}(\rew) - \lambda \mathbb{E}^{\sched}_{\cM} (\max(t-\rew,0))
	\]
	Note that in a Markov chain $\cN$, the TBPE agrees with the SMADPE if we set $t=\mathbb{E}_{\cN}(\rew)$.
	
	The main theorem  is the following.  Omitted proofs can be found in Appendix \ref{app:threshold}.
	\begin{theorem}
	\label{thm:threshold}
	Let $\cM = (S,\Act, P,\sinit, \rew,\goal)$ be an MDP satisfying Assumption \ref{ass:1} and let  $t,\lambda>0$ be rationals.
	Then, $\mathbb{E}^{\max}_{\cM} (\TBP^\lambda_t (\rew))$ and an optimal scheduler can be computed in time polynomial in the size of $\cM$ and in the numerical value of $t$.
	\end{theorem}
	The theorem follows from the following lemma:
	\begin{restatable}{lemma}{lemunfoldingsimple}
	\label{lem:unfolding}
	Given $\cM$, $t$, and $\lambda$ as in Theorem \ref{thm:threshold}, we can construct an MDP $\cM^\prime$ with reward function $\rew^\prime$ (that takes rational rewards that may be negative) and with $|S|\cdot \lceil t \rceil$ many states in time polynomial in $|S|\cdot \lceil t \rceil$
	such that 
	$
	\mathbb{E}^{\max}_{\cM}(\TBP^\lambda_t (\rew)) = \mathbb{E}^{\max}_{\cM^\prime} (\rew^\prime)$.	
	\end{restatable}
	
	\begin{proof}[Proof sketch]
	The MDP $\cM^\prime$ is an unfolding of the MDP $\cM$ that keeps track of the accumulated reward until it exceeds $t$. So, states are extended with a second component specifying the reward accumulated so far. This second component does not change anymore once it reaches $t$. For a state action pair $((s,w),\alpha)$, the new reward function is defined as $\rew' ((s,w),\alpha) = \TBP^\lambda_t (w+\rew(s,\alpha))- \TBP^\lambda_t (w)$. 
	The initial state $(\sinit,0)$ is reached via one additional new transition with reward $\TBP^\lambda_t(0)$ (which is negative).  
	\end{proof}

	While $\cM^\prime$ constructed in this proof has a rational reward function that may be  negative, the MDP $\cM^\prime$ does not contain end components.
	Hence, the maximization of the expected accumulated reward in $\cM^\prime$ can  be carried out in polynomial time \cite{Bertsekas1991AnAO} leading to Theorem \ref{thm:threshold}.
	Furthermore, memoryless deterministic schedulers for $\cM^\prime$ are sufficient for the maximization. These schedulers correspond to deterministic, finite-memory \ERMax-schedulers for $\cM$.

	\begin{remark}
	\label{rem:unfolding}
	The proof of Lemma \ref{lem:unfolding} (and Thm. \ref{thm:threshold})  works  analogously for any penalty function that penalizes outcomes below $t$: for any function $m$ such that $m(x)=x$ for $x\geq t$ that is computable in polynomial time on natural numbers, we can construct $\cM^\prime$ with a reward function $\rew^\prime$ with $|S|\cdot \lceil t \rceil$ many states in time polynomial in $|S|\cdot \lceil t \rceil$
	such that 
	$
	\mathbb{E}^{\max}_{\cM}(m (\rew)) = \mathbb{E}^{\max}_{\cM^\prime} (\rew^\prime)$ (see Appendix \ref{app:threshold}).
	Again,   $\cM^\prime$ has no end components and  the maximal expected reward in $\cM^\prime$ can be computed in  time polynomial in the size of $\cM^\prime$ \cite{Bertsekas1991AnAO}. 
		\end{remark}
	
	Finally, we show a hardness result similar as  for the MADPE. 
	
	\begin{restatable}{theorem}{thmPPthreshold}
	Given an acyclic Markov chain $\cM=(S,P,\sinit,\rew)$ and  $\vartheta, t\in \mathbb{Q}$,  deciding whether
	$
	\mathbb{E}_{\cM} ( \TBP^1_t (\rew)) \geq \vartheta
	$
	is PP-hard under polynomial-time Turing reductions.
	\end{restatable}
Note that this hardness result holds for a fixed parameter. The choice of this parameter $\lambda =1$  is arbitrary.  The proof works analogously for any positive parameter $\lambda>0$.

\section{Prototypical implementation and first experiments}
\label{sec:experiments}

	  \begin{figure}[t]
		  \begin{subfigure}[t]{0.45\textwidth}
		  \centering
    \resizebox{1\textwidth}{!}{%
    \begin{tikzpicture}
	\begin{axis}[
	scaled ticks=false,
	width=7cm,height=7cm,
	legend pos=south east,
		axis x line=center,
		axis y line=center,
		title={TBPE},
		xlabel={\#states},
		ylabel={total time [s]},
		xlabel style={below right},
		ylabel style={above},
		xmin=0,
		x tick label style={rotate=-30},
		ymin=0,
		samples=50,
		]
		\addplot+ [
		smooth,
		] coordinates {
			 (364172, 16.314641) (728172, 29.798032) (1092172, 46.457923) (1456172, 62.739622) (1820172, 86.553698) (2184172, 103.814610) (2548172, 115.758863) (2912172, 144.190125) (3276172, 160.438964) (3640172, 178.312471) (4004172, 203.361844) (4368172, 208.206041) (4732172, 238.259721) (5096172, 255.742996) (5460172, 288.362599) (5824172, 308.178525) (6188172, 346.520209) (6552172, 369.675766) (6916172, 356.527488) (7280172, 375.859814) (7644172, 420.506843) (8008172, 450.818255) (8372172, 479.715087) (8736172, 506.782537) (9100172, 533.211042) (9464172, 548.050680) (9828172, 556.085346) (10192172, 579.425285) (10556172, 550.214682) (10920172, 620.155372) (11284172, 629.182013) (11648172, 643.693275) (12012172, 644.781593) (12376172, 723.318762) (12740172, 782.481387) (13104172, 766.762190) (13468172, 797.831741) (13832172, 846.253185) (14196172, 847.982391) (14560172, 910.369410) (14924172, 915.761839) (15288172, 974.654952) (15652172, 947.835405) (16016172, 1012.419812) (16380172, 1036.674965) (16744172, 1036.700459) (17108172, 1027.873724) (17472172, 1146.769585) (17836172, 1186.243094) (18200172, 1286.438275)
		};
		\addlegendentry{N=3}
		\addplot+ [
		smooth,
		] coordinates {
			 (3172884, 127.662493) (6344884, 258.559369) (9516884, 390.120281) (12688884, 501.171399) (15860884, 676.275223) (19032884, 789.427323) (22204884, 947.116629) (25376884, 1067.317837) (28548884, 1171.995275) (31720884, 1358.117532) (34892884, 1474.764994) (38064884, 1639.419640)
		};
		\addlegendentry{N=4}
		\addplot+ [
		smooth,
		] coordinates {
			 (2733985, 84.295847) (5463885, 174.877353) (8193785, 254.727697) (10923685, 338.606079) (13653585, 424.494900) (16383485, 538.882196) (19113385, 640.328622) (21843285, 760.961085) (24573185, 812.540410) (27303085, 905.072612) (30032985, 965.206294) (32762885, 1010.867518) (35492785, 1077.796853)
		};
		\addlegendentry{N=5}
		\addplot+ [
		smooth,
		] coordinates {
			 (2387234, 45.550512) (4763794, 89.910067) (7140354, 142.068524) (9516914, 188.260565) (11893474, 261.510233) (14270034, 326.540951) (16646594, 377.538992) (19023154, 424.432298) (21399714, 470.483089) (23776274, 504.456896) (26152834, 573.357771) (28529394, 591.850729) (30905954, 622.936777) (33282514, 667.526728) (35659074, 696.733542) (38035634, 760.500449) (40412194, 800.481550) (42788754, 862.585719) (45165314, 909.402208) (47541874, 972.962292) (49918434, 1024.813779) (52294994, 1056.671422) (54671554, 1078.495180) (57048114, 1123.821595) (59424674, 1172.887974) (61801234, 1216.674977) (64177794, 1226.569383) (66554354, 1283.362471) (68930914, 1343.465717) (71307474, 1392.205142) (73684034, 1433.181150) (76060594, 1557.173916) (78437154, 1571.828488) (80813714, 1570.525641) (83190274, 1610.386967) (85566834, 1654.069574) (87943394, 1709.149576) (90319954, 1756.165319) (92696514, 1768.749286) (95073074, 1836.755114) (97449634, 1858.938921)
		};
		\addlegendentry{N=6}
		\addplot+ [
		smooth,
		] coordinates {
			 (10383259, 130.069916) (20862174, 257.190401) (31341089, 386.074433) (41820004, 504.084955) (52298919, 625.463649) (62777834, 738.282047) (73256749, 862.173280) (83735664, 960.015288) (94214579, 1093.009827) (104693494, 1225.397268) (115172409, 1363.055282) (125651324, 1484.405472) (136130239, 1616.312573) (146609154, 1776.084145)
		};
		\addlegendentry{N=7}
		\addplot+ [
		smooth,
		] coordinates {
			 (22406507, 292.956838) (35748098, 442.314397) (53691090, 603.885898) (72329350, 748.466027) (91003414, 896.992545) (109677898, 1110.267078) (128352382, 1245.100195) (147026866, 1410.967516) (165701350, 1624.614936) (184375834, 1783.756563) (203050318, 1888.187962) (221724802, 2142.773956) (240399286, 2385.431276)
		};
		\addlegendentry{N=8}
	\end{axis}
\end{tikzpicture}
    }

  \caption{The number of states of the unfolded MDPs and the time to compute the optimal TBPE for different parameter choices for the ALEP.}
   \label{fig:plot_TBPE}
  \end{subfigure}
  \hspace{12pt}
  \begin{subfigure}[t]{0.45\textwidth}
\centering
    \resizebox{1\textwidth}{!}{%
      \begin{tikzpicture}
    \begin{axis}[
    width=7cm,height=7cm,
    legend pos=north west,
        axis x line=center,
        axis y line*=left, 
        xlabel={\# states},
        ylabel={total time [s]},
        xlabel style={above left},
        ylabel style={above},
        yticklabel style={blue},
        xmin=0,
        ymin=0,
        samples=50,
        scaled y ticks=false, 
        title={MADPE},
        xmode=log,
        ymode=log
    ]
         \addplot+ [smooth,blue] coordinates {
      	(364, 0.1187) (3172, 0.9111) (27299, 12.33) (237656, 170.3)
      };\label{time}
      
      \addplot+ [smooth,blue] coordinates {
      	(272, 0.1210)
      	(528, 0.1710)
      	(1040, 0.1815)
      	(2064, 0.2790)
      	(4112, 0.4706)
      	(8208, 0.7750)
      	(16400, 1.631)
      	(32784, 3.305)
      	(65552, 6.605)
      	(131088, 14.89)
      	(262160, 39.94)
      };\label{time}
      
      \addplot+ [smooth,blue] coordinates {
      	(22656, 3.429)
      	(43136, 6.452)
      	(84096, 94.68)
      	(166016, 27.19)
      };\label{time}
      
      
      \addlegendimage{/pgfplots/refstyle=time}\addlegendentry{ALEP N=3,...,6 }
      \addlegendimage{/pgfplots/refstyle=states}\addlegendentry{RCP N=2}
      \addlegendimage{/pgfplots/refstyle=states}\addlegendentry{RCP N=4}
    \end{axis}
\end{tikzpicture}
    }

  \caption{Time to build and solve the quadratic program for the maximization of the MADPE.}
 \label{fig:plot_MADPE}
\end{subfigure}
\caption{Experimental evaluation of the algorithms for TBPE and MADPE.}
\vspace{-12pt}
\end{figure}
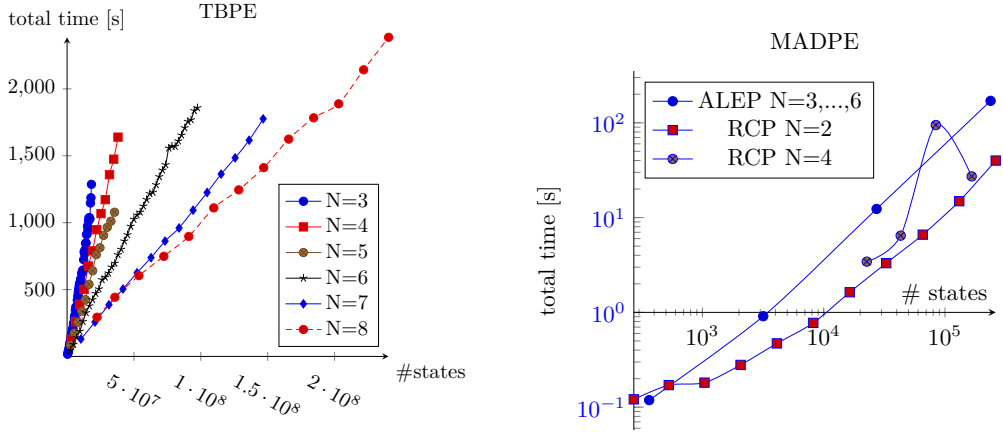

To give a prototypical proof-of-concept for the application of the MADPE and TBPE in practice, we run 
 experiments using   the model-checker PRISM  \cite{ChenForejt2013AModelChecker} and the optimization problem solver Gurobi \cite{gurobi}.
The source code for the experiments is available on github\footnote{\url{https://github.com/experiments-collection/risk-averse-stochastic-shortest-paths}}. All measurements were done on a machine running Windows 10 Pro 22H2 with an Intel Core i9-9900K CPU and 32GB RAM.
We use MDP models written in the PRISM input language (available on the PRISM website\footnote{\url{https://www.prismmodelchecker.org/}})  
for the asynchronous leader election protocol (ALEP) \cite{IR90} and, in the case of the MADPE, also for the randomized consensus protocol (RCP) \cite{AH90}. For both protocols, parameters can be chosen leading to models of different sizes and in the models for both protocols non-negative rewards are specified.

To test our algorithm for the TBPE, for each PRISM model for the ALEP with number of processes $N=3, \dots, 8$, we added a single module which implements the reward counter until reaching the threshold and  a new reward definition 
as in the construction used in the proof of Lemma \ref{lem:unfolding}. We used the penalty factor $\lambda=\frac{3}{2}$ in our examples and varied the threshold $t$. 
In Figure \ref{fig:plot_TBPE}, the sizes of the unfolded MDPs for varying values of $t$, which are proportional to  $t$, and the time needed to compute the maximal TBPE are shown. We observe that for this example the required time grows approximately linearly with the size of the unfolded MDP and consequently with the numerical value of $t$.  For the model with $N=8$, which has approximately $1.8\cdot 10^7$ many states, and $t=13$, the unfolded MDP has approximately $2.4\cdot 10^8$ many states and the computation of the optimal TBPE takes approximately $2385$ seconds.
More detailed plots for  different values for $N$ can be found in Appendix \ref{app:plots}.

To test our algorithm for the MADPE using quadratic programs, we use the  ALEP and the RCP models with various parameter choices. The parameter $\lambda$ is set to $0.4$. First we run PRISM to obtain a model representation with all  states, transitions, rewards and the maximal expected total reward from each state. Second, we run a python script which constructs all the constraints as described in Section \ref{sec:MAD} to obtain a linearly constrained program with a quadratic objective. The script uses Gurobi \cite{gurobi} to then solve the optimization problem.
The diagram in Figure \ref{fig:plot_MADPE} shows the total time for running the toolchain over the number of reachable states of each model according to PRISMs output. 
For the largest tested models with approximately $2\cdot 10^5$ many states, the maximal MADPE could be computed in less than $200$ seconds.

\section{Conclusion}
\label{sec:conclusion}

For various deviation measures, we investigated the deviation-measure-penalized expectation as risk-averse objective applied to the maximization of accumulated rewards in MDPs.
As  known from the literature, the VPE suffers from the fact that optimal schedulers have to be \ERMin-schedulers. Surprisingly, this can still be the case for the SVPE. For the MADPE, a different picture arises:
If the penalty factor $\lambda$ is at most $1/2$, optimal schedulers can be chosen to be \ERMax-schedulers. If $\lambda>1/2$, \ERMin-schedulers can be necessary.
Finally, the threshold-based penalty mechanism in the TBPE ensures that optimal schedulers are \ERMax-schedulers.
For an overview of the further results regarding computational complexity and the structure of optimal schedulers see Table  \ref{tbl:overview}.

Despite the PP-hardness results  for acyclic Markov chains, the first experimental evaluation of the two cases that ensure the existence of optimal \ERMax-schedulers, namely the TBPE in general and the MADPE for small penalty factors $\lambda \leq 1/2$,
indicates that the optimization seems to be possible in reasonable time on models of considerable size. Further experiments on the scalability of the algorithms, however, are left as future work.
 In addition, future experiments should  examine  whether the optimal schedulers for the different measures show a reasonable risk-averse behavior in case studies.

We addressed the maximization of accumulated rewards here. As we work with non-negative rewards, the case of minimization is not symmetric  and is subject to future investigations.
Finally, the studied objectives can be  transferred to other random variables such as the mean payoff, which is a further interesting direction for future work.

%
%
%
 \bibliographystyle{plainurl}
 \bibliography{references/lit2}

\clearpage
\begin{appendix}

\section{Omitted proofs of Section \ref{sec:prelim}}
\label{app:prelim}

\lemrewardbased*

\begin{proof}
For a scheduler $\sched$, a state $s$ and $w\in \mathbb{N}$,  let 
\[
\vartheta_{s,w}^{\sched}\coloneqq \mathbb{E}^{\sched}_{\cM} (\text{number of visits to state $s$ while the accumulated reward is  $w$})
\]
 be the expected frequency of visits to state $s$ with accumulated reward $w$. 
As $\goal$ is reached almost surely and is a trap state, these values are finite. So, by \cite[Lemma 2]{Piribauer2022TheVS}, there is a weight based scheduler $\tsched$ with $\vartheta_{s,w}^{\sched}= \vartheta_{s,w}^{\tsched}$
for all $s$ and $w$. The distribution of the random variable $\rew$ is  given by the values $\vartheta_{\goal, w}^{\sched}$ and $\vartheta_{\goal, w}^{\tsched}$ for $w\in \mathbb{N}$, respectively, because $\goal$ is reached almost surely and is a trap state 
so that the expected frequency of visits with accumulated reward $w$ equals the probability that the random variable $\rew$ equals $w$.
\end{proof}

\section{Omitted proofs and calculations of Section \ref{sec:MAD}}
\label{app:MAD}

\paragraph*{Calculations omitted in Example \ref{exam:minimizing}.}
We provide the calculations to determine the MAD of the scheduler $\sched$ in Example \ref{exam:minimizing}:

Observe that $\mathbb{E}^{\sched}_{\cM}(rew \mid \rew\geq 1) = \frac{1}{p}\cdot  \mathbb{E}^{\sched}_{\cM}(rew) $ as 
		\begin{align*}
		\mathbb{E}^{\sched}_{\cM}(rew) &=p \cdot  \mathbb{E}^{\sched}_{\cM}(rew \mid \rew\geq 1) +(1-p)\cdot  \mathbb{E}^{\sched}_{\cM}(rew \mid \rew =  0) \\
		&= p \cdot  \mathbb{E}^{\sched}_{\cM}(rew \mid \rew\geq 1).
		\end{align*}
		Hence, the MAD under this scheduler is 
		\begin{align*}
		&\mathbb{E}^{\sched}_{\cM} (|\rew - \mathbb{E}^{\sched}_{\cM}(rew)|)  \\
		{}={}&p \cdot  \mathbb{E}^{\sched}_{\cM} (\rew - \mathbb{E}^{\sched}_{\cM}(rew) \mid \rew\geq 1)  + (1-p) \cdot  \mathbb{E}^{\sched}_{\cM} ( \mathbb{E}^{\sched}_{\cM}(rew) -\rew \mid \rew= 0)  \\
		{}={} &   p \cdot  \mathbb{E}^{\sched}_{\cM} (\rew \mid \rew\geq 1) - p \cdot   \mathbb{E}^{\sched}_{\cM}(rew)  + (1-p) \cdot  \mathbb{E}^{\sched}_{\cM}(rew) \\
		{}={} & \frac{p}{p} \cdot  \mathbb{E}^{\sched}_{\cM}(rew) +(1-2p) \cdot   \mathbb{E}^{\sched}_{\cM}(rew)  =  2 \cdot (1-p) \cdot  \mathbb{E}^{\sched}_{\cM}(rew).
		\end{align*}

\thmMADoptimalsched*

		\begin{proof}
		 Let $\sched$ be a reward based scheduler and let 
		$\tsched$ be a memoryless deterministic scheduler with $\mathbb{E}^{\tsched}_{\cM,s}(\rew)=  \mathbb{E}^{\max}_{\cM,s}(\rew)$ for all states $s\in S$.

		Let $\psi$ be the event that a path $\pi$ has an accumulated reward of at least $k$. Then, we know that 
		\[
		E_1 := \mathbb{E}^{\sched}_{\cM}(\rew \mid \neg \psi) = \mathbb{E}^{\sched \uparrow_k \tsched }_{\cM}(\rew \mid \neg \psi) =:F_1
		\]
		and 
		\[
		E_2:= \mathbb{E}^{\sched}_{\cM}(\rew \mid  \psi) \leq \mathbb{E}^{\sched \uparrow_k \tsched }_{\cM}(\rew \mid  \psi) =: F_2.
		\]
		Furthermore,
		\[
		p:= \Pr^{\sched}_{\cM}(\psi)=  \Pr^{\sched \uparrow_k \tsched}_{\cM}(\psi).
		\]
		Note that the expected values satisfy
		$\mathbb{E}^{\sched \uparrow_k \tsched}_{\cM} (\rew) - \mathbb{E}^{\sched}_{\cM} (\rew) =  p\cdot (F_2-E_2)$.
		We now compare the MAD of both schedulers:
		\begin{align*}
		&\MAD^{\sched \uparrow_k \tsched}_{\cM}(\rew) - \MAD^{\sched}_{\cM}(\rew)  \\
		{}={}& \mathbb{E}^{\sched \uparrow_k \tsched}_{\cM}(|\rew -\mathbb{E}^{\sched \uparrow_k \tsched}_{\cM}(\rew)|) -  \mathbb{E}^{\sched }_{\cM}(|\rew -\mathbb{E}^{\sched}_{\cM}(\rew)|) \\
		{}={}&  p\cdot \left(\mathbb{E}^{\sched \uparrow_k \tsched}_{\cM}(|\rew -\mathbb{E}^{\sched \uparrow_k \tsched}_{\cM}(\rew)| \mid \psi) -  \mathbb{E}^{\sched }_{\cM}(|\rew -\mathbb{E}^{\sched}_{\cM}(\rew)|  \mid \psi ) \right) \\
		& +  (1-p) \cdot \left(\mathbb{E}^{\sched \uparrow_k \tsched}_{\cM}(|\rew -\mathbb{E}^{\sched \uparrow_k \tsched}_{\cM}(\rew)|\mid \neg \psi) -  \mathbb{E}^{\sched }_{\cM}(|\rew -\mathbb{E}^{\sched}_{\cM}(\rew)|  \mid \neg \psi )\right).
		\end{align*}
		Treating the two summands separately, we first obtain the following estimation using that if $\psi$ holds, the random variable $\rew$ takes only values larger than $\mathbb{E}^{\max}_{\cM} (\rew)$:
		\begin{align*}
		&  p\cdot \left(\mathbb{E}^{\sched \uparrow_k \tsched}_{\cM}(|\rew -\mathbb{E}^{\sched \uparrow_k \tsched}_{\cM}(\rew)| \mid \psi) -  \mathbb{E}^{\sched }_{\cM}(|\rew -\mathbb{E}^{\sched}_{\cM}(\rew)|  \mid \psi ) \right) \\
		{}={} &  p\cdot \left(\mathbb{E}^{\sched \uparrow_k \tsched}_{\cM}(\rew -\mathbb{E}^{\sched \uparrow_k \tsched}_{\cM}(\rew) \mid \psi) -  \mathbb{E}^{\sched }_{\cM}(\rew -\mathbb{E}^{\sched}_{\cM}(\rew)  \mid \psi ) \right) \\
		{}\leq{}  &  p\cdot \left(\mathbb{E}^{\sched \uparrow_k \tsched}_{\cM}(\rew -\mathbb{E}^{\sched}_{\cM}(\rew) \mid \psi) -  \mathbb{E}^{\sched }_{\cM}(\rew -\mathbb{E}^{\sched}_{\cM}(\rew)  \mid \psi ) \right) \\ 
		{}={} & p\cdot (F_2-E_2).
		\end{align*}
		Second, we observe
		\begin{align*}
		&(1-p) \cdot \left(\mathbb{E}^{\sched \uparrow_k \tsched}_{\cM}(|\rew -\mathbb{E}^{\sched \uparrow_k \tsched}_{\cM}(\rew)|\mid \neg \psi) -  \mathbb{E}^{\sched }_{\cM}(|\rew -\mathbb{E}^{\sched}_{\cM}(\rew)|  \mid \neg \psi )\right)\\
	\leq & (1-p) \cdot p \cdot (F_2-E_2)
		\end{align*}
	because the distribution of $\rew$ under the condition that $\neg \psi $ holds is the same under $\sched$ and $\sched \uparrow_k \tsched$. However, the expected values differ by $p \cdot (F_2-E_2)$. 
	For each single outcome $w$, the value $| w- \mathbb{E}^{\sched \uparrow_k \tsched}_{\cM}(\rew)|$ is hence at most $p \cdot (F_2-E_2)$ bigger than $| w- \mathbb{E}^{\sched }_{\cM}(\rew)|$ implying the stated inequality.

		Put together, this means 
		\[
		\MAD^{\sched \uparrow_k \tsched}_{\cM}(\rew) - \MAD^{\sched}_{\cM}(\rew)  \leq p\cdot (F_2-E_2) + (1-p) \cdot p \cdot (F_2-E_2) \leq 2 p \cdot  (F_2-E_2) .
		\]
		For the MADPE, this means
		\[
		\MADPE[\lambda]^{\sched \uparrow_k \tsched}_{\cM}(\rew) - \MADPE[\lambda]^{\sched}_{\cM}(\rew)  \geq   p\cdot (F_2-E_2) - \lambda \cdot 2 p \cdot  (F_2-E_2) \geq 0 
		\]
		as $\lambda \in (0,\frac{1}{2}]$. 
		  
		\end{proof}

		\lemcorrectnessconstruction*
		
		\begin{proof}
		The difference between the behavior under $\sched$ in $\cM$ and in $\cN$ is as follows:
		As soon as a reward $w\geq k$ is accumulated in $\cM$, i.e., when $(s,w)$ for this $w$ and some $s$ is reached, the reward $w+\mathbb{E}^{\max}_{\cM,s}(\rew)$ is received in $\cN$. In contrast, 
		in $\cM$, the process continues to accumulated further rewards in addition to the rewards of $w$ according to the memoryless scheduler $\tsched$. 
		If $\goal$ is reached with a reward $w^\prime < k$ in $\cM$, then the corresponding path in $\cN$ reaches $(\goal,w^\prime)$ and also receives reward $w^\prime$ in the next step. Vice versa,
		if $(\goal,w^\prime)$ is reached in $\cN$, the corresponding path in $\cM$ accumulates reward $w^\prime$.
		As $\mathbb{E}^{\tsched}_{\cM,s}(\rew) = \mathbb{E}^{\max}_{\cM,s}(\rew)$,
		this difference does not affect the expected value and so $\mathbb{E}^{\sched}_{\cM}(\rew) = \mathbb{E}^{\sched}_{\cN}(\rew)$.
		
		For the MAD, we provide the following computation. For this computation, let $A_{s,w}$ for $s\in S$ and $w\in \{k,\dots, k+\ell-1\}$ be the event that $s$ is the state that is reached as soon as the accumulated reward exceeds $k$ and that the accumulated reward at that point is $w$. Then, using the fact that $k\geq \mathbb{E}^{\sched}_{\cM}(\rew) = \mathbb{E}^{\sched}_{\cN}(\rew)$ and that the distribution of $\rew$ below $k$ are the same in $\cM$ and $\cN$ under $\sched$,
		\begin{align*}
		& \MAD^{\sched}_{\cM}(\rew) \\
		& =  \sum_{j=0}^\infty \Pr^{\sched}_{\cM}(\rew=j) \cdot | j - \mathbb{E}^{\sched}_{\cM}(\rew)| \\
		&= \sum_{j=0}^{k-1} \Pr^{\sched}_{\cN}(\rew=j) \cdot | j - \mathbb{E}^{\sched}_{\cN}(\rew)|  + \sum_{j=k}^\infty \Pr^{\sched}_{\cM}(\rew=j) \cdot | j - \mathbb{E}^{\sched}_{\cM}(\rew)| \\
		&= \sum_{j=0}^{k-1} \Pr^{\sched}_{\cN}(\rew=j) \cdot | j - \mathbb{E}^{\sched}_{\cN}(\rew)|  + \sum_{j=k}^\infty \Pr^{\sched}_{\cM}(\rew=j) \cdot ( j - \mathbb{E}^{\sched}_{\cM}(\rew)) .
		\end{align*}
		For the second summand, we get
		\begin{align*}
		 &\sum_{j=k}^\infty \Pr^{\sched}_{\cM}(\rew=j) \cdot ( j - \mathbb{E}^{\sched}_{\cM}(\rew)) \\
		&= \Pr^{\sched}_{\cM}(\rew\geq k) \cdot \mathbb{E}^{\sched}_{\cM} (\rew - \mathbb{E}^{\sched}_{\cM}(\rew) \mid \rew\geq k)  \\
		&= \sum_{s\in S, k\leq w \leq k+\ell-1}  \Pr^{\sched}_{\cM}(A_{s,w}) \cdot \mathbb{E}^{\sched}_{\cM} (\rew - \mathbb{E}^{\sched}_{\cM}(\rew) \mid A_{s,w}) \\
		&=  \sum_{s\in S, k\leq w \leq k+\ell-1}  \Pr^{\sched}_{\cM}(A_{s,w}) \cdot (w+\mathbb{E}^{\tsched}_{\cM,s} (\rew)  - \mathbb{E}^{\sched}_{\cM}(\rew) ) \\
		&=  \sum_{s\in S, k\leq w \leq k+\ell-1}  \Pr^{\sched}_{\cN}(\Diamond (s,w)) \cdot (w+\mathbb{E}^{\max}_{\cM,s} (\rew)  - \mathbb{E}^{\sched}_{\cN}(\rew) ) \\
		&= \Pr^{\sched}_{\cN}(\rew \geq k) \cdot (\rew -  \mathbb{E}^{\sched}_{\cN}(\rew) \mid \rew\geq k).
		\end{align*}
		Putting this together, we get 
		\begin{align*}
		&\MAD^{\sched}_{\cM}(\rew)  \\
		& = \sum_{j=0}^{k-1} \Pr^{\sched}_{\cN}(\rew=j) \cdot | j - \mathbb{E}^{\sched}_{\cN}(\rew)|  + \Pr^{\sched}_{\cN}(\rew \geq k) \cdot (\rew -  \mathbb{E}^{\sched}_{\cN}(\rew) \mid \rew\geq k)\\
		& =\MAD^{\sched}_{\cN}(\rew),
		\end{align*}
		which also shows that the MADPE is the same in $\cM$ and $\cN$ under $\sched$. 
		  
		\end{proof}

		\lemfrequency*

		\begin{proof}
		As shown in \cite[Theorem 4.7]{Kallenberg}, for any solution vector $x$ to the  system of constraints (\ref{eq:c1}) - (\ref{eq:c2}), there is a scheduler $\sched$ for $\cN$ such that
		the expected number of times that $\alpha$ is chosen in state $(s,w)$ under $\sched$ is equal to $x_{s,w,\alpha}$ -- and vice versa.
		As states $(s,w)$ with $s=\goal$ or $w\geq k$ can be visited at most once during a run and only $\tau$ is enabled in these states, these expected number of times that action $\tau$ is chosen in these states is exactly the probability that they are visited. The claim for the expected value follows directly from the definition of $\rew'$.
		  
		\end{proof}

		\thmPPMAD*
		
		\begin{proof}
		Let $\cM=(S,P,\sinit, \rew)$ be an acyclic Markov chain. As we assume all transition probabilities to be rational, the  product $L$ of all the denominators of the transition probabilities has the property that the probability of any path is a multiple of $1/L$.
		Furthermore, $L$ can be computed in polynomial time. We can also compute the value $E:=\mathbb{E}_{\cM}(\rew)$ in polynomial time as well as the maximal reward $K$ of a path.
		Now, $E$ is a multiple of $1/L$. So, for each path $\pi$, the value $|\rew(\pi) - E|$ is a multiple of $1/L$.
		So, the MAD $\MAD_{\cM} (\rew)= \sum_{\text{paths }\pi} P(\pi)\cdot |\rew(\pi) - E|$ is a multiple of $1/L^2$.
		As $\MAD_{\cM}$ is furthermore at most $K$, there are only $L^2\cdot K$ many possible values for $\MAD_{\cM}$.
		If we have access to an oracle for the threshold problem for the MAD, we can hence compute the exact value $\MAD_{\cM}$ via a binary search with $\lceil \log(L^2\cdot K) \rceil$ many calls to the oracle.
		As $L$ and $K$ can be computed in polynomial time, their logarithms are polynomial in the size of the Markov chain $\cM$.
		
		Knowing that we can use an oracle for the threshold problem to exactly compute the MAD in an acyclic Markov chain in polynomial time, we will now prove the PP-hardness of this threshold problem.
		We reduce from the following  problem that is shown to be PP-hard in \cite{HaaseK15}: 
		Given an acyclic Markov chain $\cM=(S,P,\sinit, \rew)$, and a natural number $t$, decide whether $\Pr_{\cM}(\rew>t) \geq 1/2$.
		
		So, let such a Markov chain $M$ and a natural number $t$ be given. Note that we can say that the reward function maps states to rewards in Markov chains as there are no actions.
		 Define again $E:= \mathbb{E}_{\cM}(\rew)$.
		First, we assume that $t\geq E$. The case that $t<E$ can be treated very similarly and will be described below.
		Now, we construct two acyclic Markov chains $\cM_1$ and $\cM_2$ as follows.
		In both Markov chains, we add a new initial state $\sinit^\prime$ with reward $0$ to $\cM$ from which the original initial state $\sinit$ is reached with probability $1/2$. Also with probability $1/2$, a new trap state $s$ is reached. 
		In $\cM_1$, the state $s$ has reward $2t-E$. So, $\mathbb{E}_{\cM_1}(\rew) = 1/2 \cdot E + 1/2\cdot (2t-E) = t$.
		In $\cM_2$, the state $s$ has reward $2t-E+1$. So, $\mathbb{E}_{\cM_2}(\rew) = t+1/2$.
		
		Now, we use the procedure described above with calls to an oracle for the threshold problem for MAD to compute the values $\MAD_{\cM_1}$ and $\MAD_{\cM_2}$.
		We know that 
		\[
		\MAD_{\cM_1}(\rew) = 1/2 |2t-E - t| + 1/2\cdot \sum_{w=0}^K \Pr_{\cM}(\rew = w) \cdot  |w-t|. 
		\]
		So, we can derive the value $\sum_{w=0}^K \Pr_{\cM}(\rew = w) \cdot  |w-t|$. 
		Likewise, in $\cM_2$, we get 
		\[
		\MAD_{\cM_2}(\rew) = 1/2 |2t-E+1 - t -1/2| + 1/2\cdot \sum_{w=0}^K \Pr_{\cM}(\rew = w) \cdot  |w-t-1/2|
		\]
		and can hence obtain the value $\sum_{w=0}^K \Pr_{\cM}(\rew = w) \cdot  |w-t-1/2|$.
		Now, observe that 
		\begin{align*}
			&	\sum_{w=0}^K \Pr_{\cM}(\rew = w) \cdot  |w-t|  - \sum_{w=0}^K \Pr_{\cM}(\rew = w) \cdot  |w-t-1/2| \\
			{}={} & \sum_{w=0}^t \Pr_{\cM}(\rew = w) \cdot  (t-w) + \sum_{w=t+1}^K \Pr_{\cM}(\rew = w) \cdot  (w-t)  \\
			&- \left(  \sum_{w=0}^t \Pr_{\cM}(\rew = w) \cdot  (t+1/2-w) + \sum_{w=t+1}^K \Pr_{\cM}(\rew = w) \cdot  (w-t-1/2) \right) \\
			{}={} &\sum_{w=0}^t \Pr_{\cM}(\rew = w) \cdot  (t-w) -\sum_{w=0}^t \Pr_{\cM}(\rew = w) \cdot  (t+1/2-w)  \\
			&+  \sum_{w=t+1}^K \Pr_{\cM}(\rew = w) \cdot  (w-t) - \sum_{w=t+1}^K \Pr_{\cM}(\rew = w) \cdot  (w-t-1/2) \\
			{}={} &\sum_{w=0}^t \Pr_{\cM}(\rew = w) \cdot (-1/2)  + \sum_{w=t+1}^K \Pr_{\cM}(\rew = w) \cdot 1/2 \\
			{}={} & 1/2\cdot (\Pr_{\cM}(\rew> t) - \Pr_{\cM}(\rew\leq t) ) \\
			{}={} & 1/2\cdot (\Pr_{\cM}(\rew> t) - (1-\Pr_{\cM}(\rew> t))) = \Pr_{\cM}(\rew> t) - 1/2.
		\end{align*}
		So, these two values allow us to compute $\Pr_{\cM}(\rew> t) $. Hence, we can answer whether $\Pr_{\cM}(\rew> t)\geq 1/2$.
		
		To complete the proof let us briefly sketch the case that $t<E$. 
		We can conduct a very similar construction as for $\cM_1$ and $\cM_2$ to ensure that the expected values in the resulting Markov chains are $t$ and $t+1/2$.
		E.g., we can construct $\cM_1^\prime$ by adding a new initial state to $\cM$ from which the original initial state is reached with probability $t/E$ and a terminal state with reward $0$ is reached with probability $1-t/E$.
		For $\cM_2$, we use probabilities $(t+1/2)/E$ and $1-(t+1/2)/E$.
		The remaining argument with slightly modified calculations can be carried out in exactly the same way.
		  
		\end{proof}

		\section{Computations omitted in Section \ref{sec:semivariance}}
		\label{app:SMAD}

		\noindent \textbf{Computation for Section \ref{sec:semiMAD}.}
		 For a random variable $X$ with expected value $E$ only taking natural numbers, we have
		\begin{align*}
		&\MAD(X) - 2 \SMAD(X)\\
		 &= \sum_{k=0}^{\lfloor E \rfloor} \Pr(X=k)(E-k) + \sum_{k=\lceil E \rceil}^{\infty } \Pr(X=k)(k-E) - 2  \sum_{k=0}^{\lfloor E \rfloor } \Pr(X=k)(E-k) \\
		&= \sum_{k=\lceil E \rceil}^{\infty } \Pr(X=k)(k-E) - \sum_{k=0}^{\lfloor E \rfloor } \Pr(X=k)(E-k) \\
		&= \sum_{k=0}^{\infty }\Pr(X=k)(k-E) = \left( \sum_{k=0}^{\infty }\Pr(X=k)k\right) - E = 0.
		\end{align*}

		\section{Omitted proofs of Section \ref{sec:threshold}}
		\label{app:threshold}
		
		We are going to prove the following lemma from the main paper that we restate here in a more general form as described in Remark \ref{rem:unfolding}.
		
		\lemunfoldingsimple*
		
		For the more general form, let $\cM$ and $t$ be as above and let 
		  $m$ be a function with $m(x) = x$ for $x\geq t$ that is computable in polynomial time on natural numbers. Note that $\TBP^\lambda_t$ is such a function.

		\begin{restatable}{lemma}{lemunfolding}
	Given $\cM$, $t$, and $m$ as above, we can construct an MDP $\cM^\prime$ with reward function $\rew^\prime$ and with $|S|\cdot \lceil t \rceil$ many states in polynomial time
	such that 
	\[
	\mathbb{E}^{\max}_{\cM}(m\circ \rew) = \mathbb{E}^{\max}_{\cM^\prime} (\rew^\prime).
	\]
	\end{restatable}

		\begin{proof}
	In the sequel, we provide an unfolding $\cM^\prime=(S^\prime,\Act^\prime,P^\prime, \sinit^\prime, \rew^\prime, \goal^\prime)$ of $\cM$ with a new reward function $\rew^\prime$.
	
	The state space $S^\prime$ is given by $S\times \{0,\dots, \lceil t \rceil \}\cup\{\sinit^\prime\}$ where $\sinit^\prime$ is a new initial state. The set of actions is extended by one new action $\tau$, so $\Act^\prime = \Act \cup \{\tau\}$.
	For $(s,w)$ and $(t,v)$ in $S^\prime$ and $\alpha\in \Act$, we define
	\[
	P^\prime((s,w),\alpha, (r,v)) = \begin{cases}
	P(s,\alpha,r) & \text{ if $v=w+\rew(s,\alpha)$} \\
	& \text{ or $v=\lceil t \rceil$ and $w+\rew(s,\alpha)\geq t$,} \\
	0 & \text{ otherwise.}
	\end{cases}
	\]
	In the initial state $\sinit$ only $\tau$ is enabled which leads to $(\sinit,0)$ with probability $1$.
	 The new \emph{rational} reward function $\rew^\prime$ is given by
	\[
	\rew^\prime((s,w),\alpha) = m(w+\rew(s,\alpha))-m(w)
	\]
	for $(s,w)\in S^\prime$ and $\alpha\in \Act$.
	Note that this means that $\rew^\prime((s,\lceil t \rceil), \alpha) = \rew(s,\alpha)$ for all $s$ and $\alpha$.
	This is also the reason why we do not have to track the reward accumulated so far above $\lceil t \rceil$.
	Further $\rew^\prime (\sinit^\prime, \tau) = m(0)$.
	
	Finally, in order to obtain only one trap state $\goal^\prime$, we could collapse all states of the form $(\goal,w)$ to a single trap state.
	For convenience, however, we keep this set of trap states unchanged in $\cM^\prime$.
	
	Clearly, there is now a 1-to-1-correspondence between paths in $\cM$ and $\cM^\prime$. 
	For a path $s_0s_1\dots$ in $\cM$, we can obtain a path in $\cM^\prime$ by adding the initial state $\sinit^\prime$ in the beginning and afterwards adding the reward that has been accumulated so far up to $\lceil t \rceil$ as a second component to the states.
	Vice versa, for a path $\sinit^\prime (s_0,w_0)(s_1,w_1)\dots$ in $\cM^\prime$, we obtain the corresponding path in $\cM$ by projection onto the states $s_0s_1\dots$.
	For a path $\pi$ in $\cM$ and the corresponding path $\pi^\prime$ in $\cM^\prime$, we observe that the definition of $\rew^\prime$ ensures that
	\[
	m \circ \rew(\pi) = \rew^\prime (\pi^\prime).
	\]

	As the same actions are enabled in a state $(s,w)$ in $\cM^\prime$ and in state $s$ in $\cM$, there is furthermore a 1-to-1-correspondence between schedulers for $\cM$ and $\cM^\prime$.
	For each scheduler $\sched$ viewed as a scheduler for both MDPs, we get
	\[
	\mathbb{E}^{\sched}_{\cM}(m\circ \rew) = \mathbb{E}^{\sched}_{\cM^\prime}( \rew^\prime).
	\]
	  
	\end{proof}

	\thmPPthreshold*
	
	\begin{proof}
	We first show that for any $t$ we can write the probability $\mathbb{Pr}(\rew(\pi) \geq t)$ that a random path $\pi$ as the difference of expected values of the threshold-based penalized expectation:
	\begin{equation*}
		\mathbb{Pr}(\rew(\pi) \geq t) = \mathbb{E}_{\cM} ( \mathrm{crinkle}_{t}^{2} \circ \rew) - \mathbb{E}_{\cM} ( \mathrm{crinkle}_{t-1}^{2} \circ \rew)
	\end{equation*}
	To show the equation above, we first note that since $\mathrm{crinkle}_t^2$ is defined
	\begin{equation}
		\mathrm{crinkle}_t^2 (x) = \begin{cases}
			2x & \text{(if $x < t$)}\\
			x+t & \text{(if $x \geq t$)}
		\end{cases}\text{,}
	\end{equation}
	we obtain $\mathrm{crinkle}_t^2 (x) - \mathrm{crinkle}_{t-1}^2 (x) = 0$ for $x\leq t-1$ while $\mathrm{crinkle}_t^2 (x) - \mathrm{crinkle}_{t-1}^2 (x) = 1$ for $x\geq t$.
	Hence the probability $\mathbb{Pr}(\rew(\pi) \geq t)$ can be seen as
	\begin{align*}
		&\mathbb{Pr}(\rew(\pi) \geq t)\\
		=& \sum_{\pi \in \mathrm{Paths}^{\mathrm{max}(\cM)}}{P(\pi) \cdot \big(
		 		\mathrm{crinkle}_t^2 (\rew(\pi)) - \mathrm{crinkle}_{t-1}^2 (\rew(\pi))\big)}\\
	 	=& \sum_{\pi \in \mathrm{Paths}^{\mathrm{max}(\cM)}}{P(\pi) \cdot 
	 		\mathrm{crinkle}_t^2 (\rew(\pi))} -\sum_{\pi \in \mathrm{Paths}^{\mathrm{max}(\cM)}}{P(\pi) \cdot \mathrm{crinkle}_{t-1}^2 (\rew(\pi))}\\
 		=& \mathbb{E}_{\cM}(\mathrm{crinkle}_t^2 \circ \rew) - \mathbb{E}_{\cM}(\mathrm{crinkle}_{t-1}^2 \circ \rew)\text{.}
	\end{align*}
	The remainder of the proof is now analogous to the proof of Theorem \ref{thm:PP-hard_MAD}:
	In acyclic Markov chains, we can compute a number $L$ in polynomial time such that the probability of all paths is a multiple of $1/L$, 
	namely we can choose $L$ to be the product of all denominators of all transition probabilities.
	Furthermore the expected value of $\mathrm{crinkle}_t^2 \circ \rew$ is at most twice the maximal reward $R$ of a path in the Markov chain. So, there are only 
	$2RL$ many possible values for this expected value. Using the threshold problem for this expected value as an oracle, 
	we can compute the exact expected value as in the proof of Theorem  \ref{thm:PP-hard_MAD} via a binary search in polynomial time.
	
	This completes the  polynomial-time Turing reduction showing that the threshold problem is PP-hard under such reductions. 
	  
	\end{proof}

\section{Experimental evaluation}
\label{app:plots}

In the sequel, the number of states of the unfolded MDPs as well as the time to compute the maximal TBPE as described in Section \ref{sec:experiments} are depicted for the Asynchronous Leader Election Protocol with parameter
$N=3,\dots, 8$ and varying values of the paramter $t$. The number of states of the unfolded MDPs grows linearly in $t$ as expected. Interestingly, also the required times seem to grow linearly in $t$.

\begin{tikzpicture}
    \begin{axis}[
    width=0.7\textwidth,
    legend pos=north west,
        axis x line=center,
        axis y line*=left, 
        xlabel={parameter $t$ },
        ylabel={total time [s]},
        xlabel style={above left},
        ylabel style={above},
        yticklabel style={blue},
         xticklabel style={rotate=-90},
        xmin=0,
        ymin=0,
        samples=50,
        scaled x ticks=false,
        scaled y ticks=false, 
        title={$N=3$},
    ]
        \addplot+ [smooth,blue] coordinates {
            (1000, 16.314641) (2000, 29.798032) (3000, 46.457923) (4000, 62.739622) (5000, 86.553698) (6000, 103.814610) (7000, 115.758863) (8000, 144.190125) (9000, 160.438964) (10000, 178.312471) (11000, 203.361844) (12000, 208.206041) (13000, 238.259721) (14000, 255.742996) (15000, 288.362599) (16000, 308.178525) (17000, 346.520209) (18000, 369.675766) (19000, 356.527488) (20000, 375.859814) (21000, 420.506843) (22000, 450.818255) (23000, 479.715087) (24000, 506.782537) (25000, 533.211042) (26000, 548.050680) (27000, 556.085346) (28000, 579.425285) (29000, 550.214682) (30000, 620.155372) (31000, 629.182013) (32000, 643.693275) (33000, 644.781593) (34000, 723.318762) (35000, 782.481387) (36000, 766.762190) (37000, 797.831741) (38000, 846.253185) (39000, 847.982391) (40000, 910.369410) (41000, 915.761839) (42000, 974.654952) (43000, 947.835405) (44000, 1012.419812) (45000, 1036.674965) (46000, 1036.700459) (47000, 1027.873724) (48000, 1146.769585) (49000, 1186.243094) (50000, 1286.438275) (100000, 2732.295688) (150000, 4867.411873) (200000, 6713.309162)
        };\label{time}
    \end{axis}

    \begin{axis}[
    width=0.7\textwidth,
    legend pos=north west,
        axis x line=center,
        axis y line*=right, 
        ylabel={\#states },
        xlabel style={below right},
        ylabel style={above right},
        xmin=0,
        ymin=0,
        samples=50,
        scaled y ticks=false, 
        axis y line=right, 
        xtick=\empty, 
        yticklabel style={red}, 
    ]
    
    \addplot+ [
		smooth,
		] coordinates { (0,0) };

        \addplot+ [smooth, red] coordinates {
             (1000, 364172.00) (2000, 728172.00) (3000, 1092172.00) (4000, 1456172.00) (5000, 1820172.00) (6000, 2184172.00) (7000, 2548172.00) (8000, 2912172.00) (9000, 3276172.00) (10000, 3640172.00) (11000, 4004172.00) (12000, 4368172.00) (13000, 4732172.00) (14000, 5096172.00) (15000, 5460172.00) (16000, 5824172.00) (17000, 6188172.00) (18000, 6552172.00) (19000, 6916172.00) (20000, 7280172.00) (21000, 7644172.00) (22000, 8008172.00) (23000, 8372172.00) (24000, 8736172.00) (25000, 9100172.00) (26000, 9464172.00) (27000, 9828172.00) (28000, 10192172.00) (29000, 10556172.00) (30000, 10920172.00) (31000, 11284172.00) (32000, 11648172.00) (33000, 12012172.00) (34000, 12376172.00) (35000, 12740172.00) (36000, 13104172.00) (37000, 13468172.00) (38000, 13832172.00) (39000, 14196172.00) (40000, 14560172.00) (41000, 14924172.00) (42000, 15288172.00) (43000, 15652172.00) (44000, 16016172.00) (45000, 16380172.00) (46000, 16744172.00) (47000, 17108172.00) (48000, 17472172.00) (49000, 17836172.00) (50000, 18200172.00) (100000, 36400172.00) (150000, 54600172.00) (200000, 72800172.00)
        };\label{states}
         \addlegendimage{/pgfplots/refstyle=time}\addlegendentry{total time [s]}
         \addlegendimage{/pgfplots/refstyle=states}\addlegendentry{\# states}
              \end{axis}
\end{tikzpicture}

%
%
%

\begin{tikzpicture}
    \begin{axis}[
    width=0.7\textwidth,
    legend pos=north west,
        axis x line=center,
        axis y line*=left, 
        xlabel={parameter $t$ },
        ylabel={total time [s]},
        xlabel style={above left},
        ylabel style={above},
        yticklabel style={blue},
        xmin=0,
        ymin=0,
        samples=50,
         scaled x ticks=false,
        scaled y ticks=false, 
        title={$N=4$},
    ]
        \addplot+ [smooth,blue] coordinates {
           (1000, 127.662493) (2000, 258.559369) (3000, 390.120281) (4000, 501.171399) (5000, 676.275223) (6000, 789.427323) (7000, 947.116629) (8000, 1067.317837) (9000, 1171.995275) (10000, 1358.117532) (11000, 1474.764994) (12000, 1639.419640)
        };\label{time}
    \end{axis}

    \begin{axis}[
    width=0.7\textwidth,
    legend pos=north west,
        axis x line=center,
        axis y line*=right, 
        ylabel={\#states },
        xlabel style={below right},
        ylabel style={above right},
        xmin=0,
        ymin=0,
        samples=50,
        scaled y ticks=false, 
        axis y line=right, 
        xtick=\empty, 
        yticklabel style={red}, 
    ]
    
    \addplot+ [
		smooth,
		] coordinates { (0,0) };

        \addplot+ [smooth, red,y filter/.code={\pgfmathparse{\pgfmathresult*10000.}\pgfmathresult}
        ] coordinates {
             (1000, 317.288400) (2000, 634.488400) (3000, 951.688400) (4000, 1268.888400) (5000, 1586.088400) (6000, 1903.288400) (7000, 2220.488400) (8000, 2537.688400) (9000, 2854.888400) (10000, 3172.088400) (11000, 3489.288400) (12000, 3806.488400)
        };\label{states}
         \addlegendimage{/pgfplots/refstyle=time}\addlegendentry{total time [s]}
         \addlegendimage{/pgfplots/refstyle=states}\addlegendentry{\# states}
              \end{axis}
\end{tikzpicture}

%
%
%

\begin{tikzpicture}
    \begin{axis}[
    width=0.7\textwidth,
    legend pos=north west,
        axis x line=center,
        axis y line*=left, 
        xlabel={parameter $t$ },
        ylabel={total time [s]},
        xlabel style={above left},
        ylabel style={above},
        yticklabel style={blue},
        xmin=0,
        ymin=0,
        samples=50,
        scaled y ticks=false, 
        title={$N=5$},
    ]
        \addplot+ [smooth,blue] coordinates {
           (100, 84.295847) (200, 174.877353) (300, 254.727697) (400, 338.606079) (500, 424.494900) (600, 538.882196) (700, 640.328622) (800, 760.961085) (900, 812.540410) (1000, 905.072612) (1100, 965.206294) (1200, 1010.867518) (1300, 1077.796853)
        };\label{time}
    \end{axis}

    \begin{axis}[
    width=0.7\textwidth,
    legend pos=north west,
        axis x line=center,
        axis y line*=right, 
        ylabel={\#states },
        xlabel style={below right},
        ylabel style={above right},
        xmin=0,
        ymin=0,
        samples=50,
        scaled y ticks=false, 
        axis y line=right, 
        xtick=\empty, 
        yticklabel style={red}, 
    ]
    
    \addplot+ [
		smooth,
		] coordinates { (0,0) };

        \addplot+ [smooth, red,y filter/.code={\pgfmathparse{\pgfmathresult*10000.}\pgfmathresult}
        ] coordinates {
             (100, 273.398500) (200, 546.388500) (300, 819.378500) (400, 1092.368500) (500, 1365.358500) (600, 1638.348500) (700, 1911.338500) (800, 2184.328500) (900, 2457.318500) (1000, 2730.308500) (1100, 3003.298500) (1200, 3276.288500) (1300, 3549.278500)
        };\label{states}
         \addlegendimage{/pgfplots/refstyle=time}\addlegendentry{total time [s]}
         \addlegendimage{/pgfplots/refstyle=states}\addlegendentry{\# states}
              \end{axis}
\end{tikzpicture}

\begin{tikzpicture}
    \begin{axis}[
    width=0.7\textwidth,
    legend pos=north west,
        axis x line=center,
        axis y line*=left, 
        xlabel={parameter $t$ },
        ylabel={total time [s]},
        xlabel style={above left},
        ylabel style={above},
        yticklabel style={blue},
        xmin=0,
        ymin=0,
        samples=50,
        scaled y ticks=false, 
        title={$N=6$},
    ]
        \addplot+ [smooth,blue] coordinates {
           (10, 45.550512) (20, 89.910067) (30, 142.068524) (40, 188.260565) (50, 261.510233) (60, 326.540951) (70, 377.538992) (80, 424.432298) (90, 470.483089) (100, 504.456896) (110, 573.357771) (120, 591.850729) (130, 622.936777) (140, 667.526728) (150, 696.733542) (160, 760.500449) (170, 800.481550) (180, 862.585719) (190, 909.402208) (200, 972.962292) (210, 1024.813779) (220, 1056.671422) (230, 1078.495180) (240, 1123.821595) (250, 1172.887974) (260, 1216.674977) (270, 1226.569383) (280, 1283.362471) (290, 1343.465717) (300, 1392.205142) (310, 1433.181150) (320, 1557.173916) (330, 1571.828488) (340, 1570.525641) (350, 1610.386967) (360, 1654.069574) (370, 1709.149576) (380, 1756.165319) (390, 1768.749286) (400, 1836.755114) (410, 1858.938921)
                   };\label{time}
    \end{axis}

    \begin{axis}[
    width=0.7\textwidth,
    legend pos=north west,
        axis x line=center,
        axis y line*=right, 
        ylabel={\#states },
        xlabel style={below right},
        ylabel style={above right},
        xmin=0,
        ymin=0,
        samples=50,
        scaled y ticks=false, 
        axis y line=right, 
        xtick=\empty, 
        yticklabel style={red}, 
    ]
    
    \addplot+ [
		smooth,
		] coordinates { (0,0) };

        \addplot+ [smooth, red,y filter/.code={\pgfmathparse{\pgfmathresult*10000.}\pgfmathresult}
        ] coordinates {
             (10, 238.723400) (20, 476.379400) (30, 714.035400) (40, 951.691400) (50, 1189.347400) (60, 1427.003400) (70, 1664.659400) (80, 1902.315400) (90, 2139.971400) (100, 2377.627400) (110, 2615.283400) (120, 2852.939400) (130, 3090.595400) (140, 3328.251400) (150, 3565.907400) (160, 3803.563400) (170, 4041.219400) (180, 4278.875400) (190, 4516.531400) (200, 4754.187400) (210, 4991.843400) (220, 5229.499400) (230, 5467.155400) (240, 5704.811400) (250, 5942.467400) (260, 6180.123400) (270, 6417.779400) (280, 6655.435400) (290, 6893.091400) (300, 7130.747400) (310, 7368.403400) (320, 7606.059400) (330, 7843.715400) (340, 8081.371400) (350, 8319.027400) (360, 8556.683400) (370, 8794.339400) (380, 9031.995400) (390, 9269.651400) (400, 9507.307400) (410, 9744.963400)
        };\label{states}
         \addlegendimage{/pgfplots/refstyle=time}\addlegendentry{total time [s]}
         \addlegendimage{/pgfplots/refstyle=states}\addlegendentry{\# states}
              \end{axis}
\end{tikzpicture}

%
%
%
%
%
%

\begin{tikzpicture}
    \begin{axis}[
    width=0.7\textwidth,
    legend pos=north west,
        axis x line=center,
        axis y line*=left, 
        xlabel={parameter $t$ },
        ylabel={total time [s]},
        xlabel style={above left},
        ylabel style={above},
        yticklabel style={blue},
        xmin=0,
        ymin=0,
        samples=50,
        scaled y ticks=false, 
        title={$N=7$},
    ]
        \addplot+ [smooth,blue,y filter/.code={\pgfmathparse{\pgfmathresult*1000.}\pgfmathresult}] coordinates {
			 (5, 0.130070) (10, 0.257190) (15, 0.386074) (20, 0.504085) (25, 0.625464) (30, 0.738282) (35, 0.862173) (40, 0.960015) (45, 1.093010) (50, 1.225397) (55, 1.363055) (60, 1.484405) (65, 1.616313) (70, 1.776084)
                   };\label{time}
    \end{axis}

    \begin{axis}[
    width=0.7\textwidth,
    legend pos=north west,
        axis x line=center,
        axis y line*=right, 
        ylabel={\#states },
        xlabel style={below right},
        ylabel style={above right},
        xmin=0,
        ymin=0,
        samples=50,
        scaled y ticks=false, 
        axis y line=right, 
        xtick=\empty, 
        yticklabel style={red}, 
    ]
    
    \addplot+ [
		smooth,
		] coordinates { (0,0) };

        \addplot+ [smooth, red,
        y filter/.code={\pgfmathparse{\pgfmathresult*10000000.}\pgfmathresult}
        ] coordinates {
			 (5, 1.038326) (10, 2.086217) (15, 3.134109) (20, 4.182000) (25, 5.229892) (30, 6.277783) (35, 7.325675) (40, 8.373566) (45, 9.421458) (50, 10.469349) (55, 11.517241) (60, 12.565132) (65, 13.613024) (70, 14.660915)
        };\label{states}
         \addlegendimage{/pgfplots/refstyle=time}\addlegendentry{total time [s]}
         \addlegendimage{/pgfplots/refstyle=states}\addlegendentry{\# states}
              \end{axis}
\end{tikzpicture}

%
%
%

\begin{tikzpicture}
    \begin{axis}[
    width=0.7\textwidth,
    legend pos=north west,
        axis x line=center,
        axis y line*=left, 
        xlabel={parameter $t$ },
        ylabel={total time [s]},
        xlabel style={above left},
        ylabel style={above},
        yticklabel style={blue},
        xmin=0,
        ymin=0,
        samples=50,
        scaled y ticks=false, 
        title={$N=8$},
    ]
        \addplot+ [smooth,blue,y filter/.code={\pgfmathparse{\pgfmathresult*1000.}\pgfmathresult}] coordinates {
			 (1, 0.292957) (2, 0.442314) (3, 0.603886) (4, 0.748466) (5, 0.896993) (6, 1.110267) (7, 1.245100) (8, 1.410968) (9, 1.624615) (10, 1.783757) (11, 1.888188) (12, 2.142774) (13, 2.385431)
                   };\label{time}
    \end{axis}

    \begin{axis}[
    width=0.7\textwidth,
    legend pos=north west,
        axis x line=center,
        axis y line*=right, 
        ylabel={\#states },
        xlabel style={below right},
        ylabel style={above right},
        xmin=0,
        ymin=0,
        samples=50,
        scaled y ticks=false, 
        axis y line=right, 
        xtick=\empty, 
        yticklabel style={red}, 
    ]
    
    \addplot+ [
		smooth,
		] coordinates { (0,0) };

        \addplot+ [smooth, red,
        y filter/.code={\pgfmathparse{\pgfmathresult*10000000.}\pgfmathresult}
        ] coordinates {
			 (1, 2.240651) (2, 3.574810) (3, 5.369109) (4, 7.232935) (5, 9.100341) (6, 10.967790) (7, 12.835238) (8, 14.702687) (9, 16.570135) (10, 18.437583) (11, 20.305032) (12, 22.172480) (13, 24.039929)
        };\label{states}
         \addlegendimage{/pgfplots/refstyle=time}\addlegendentry{total time [s]}
         \addlegendimage{/pgfplots/refstyle=states}\addlegendentry{\# states}
              \end{axis}
\end{tikzpicture}

\end{appendix}

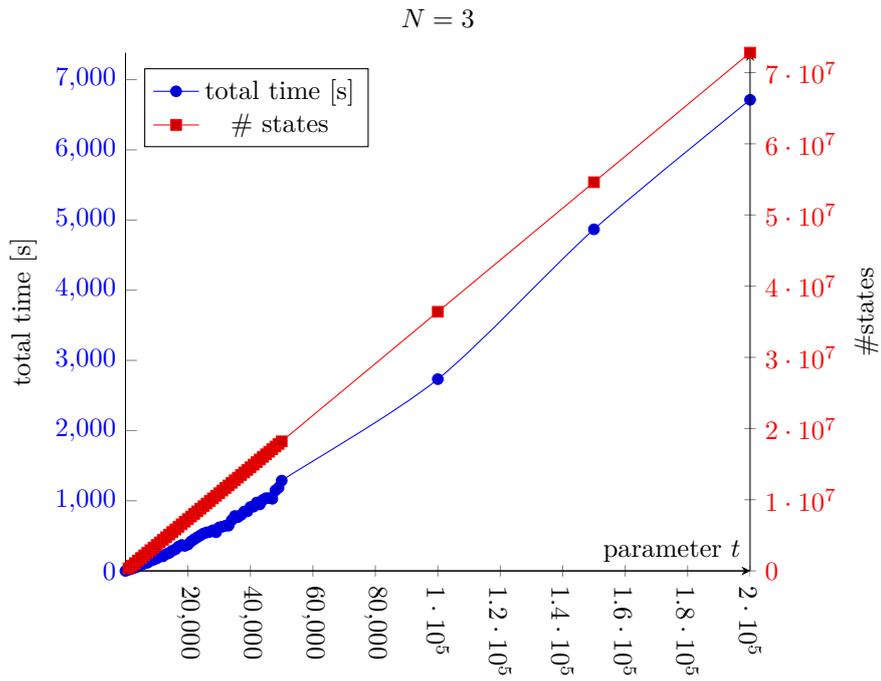
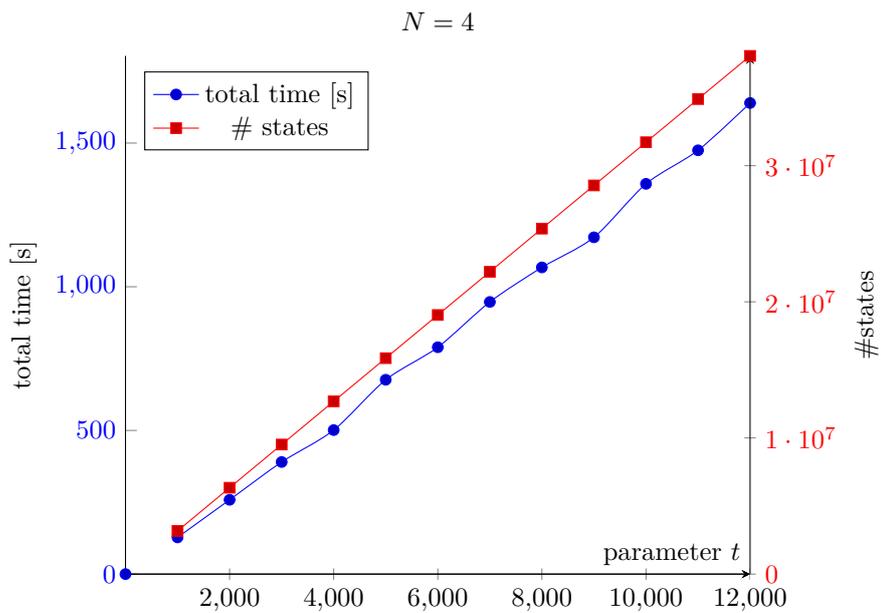
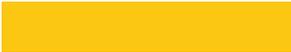
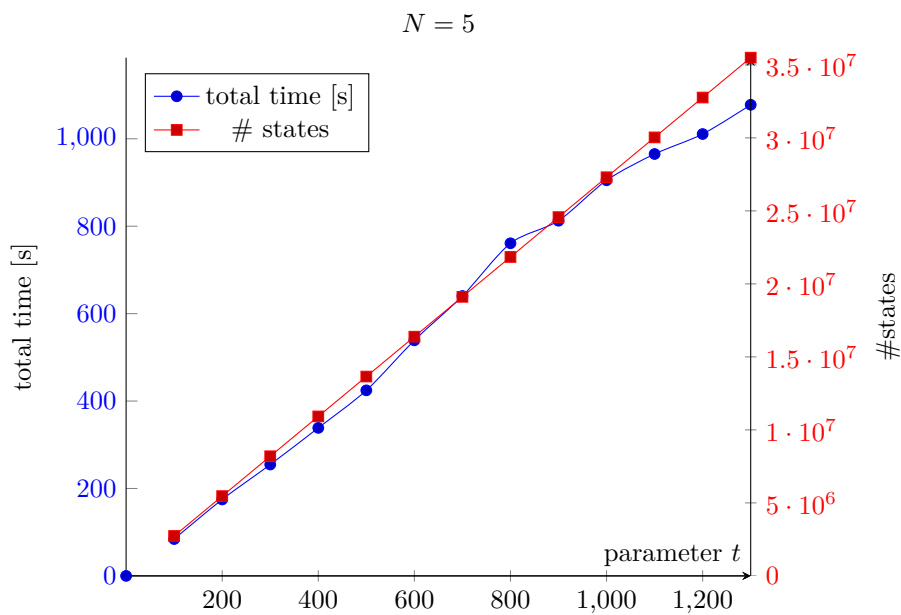
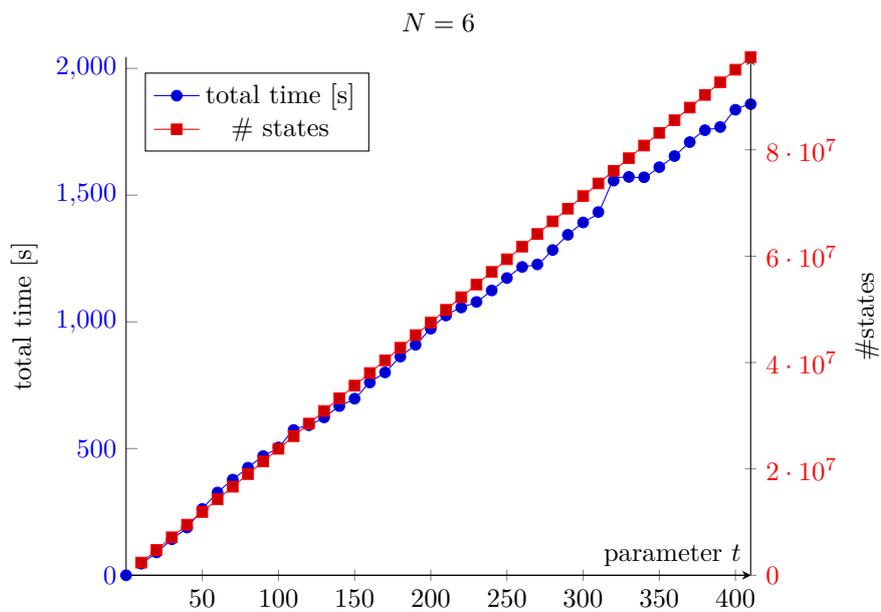
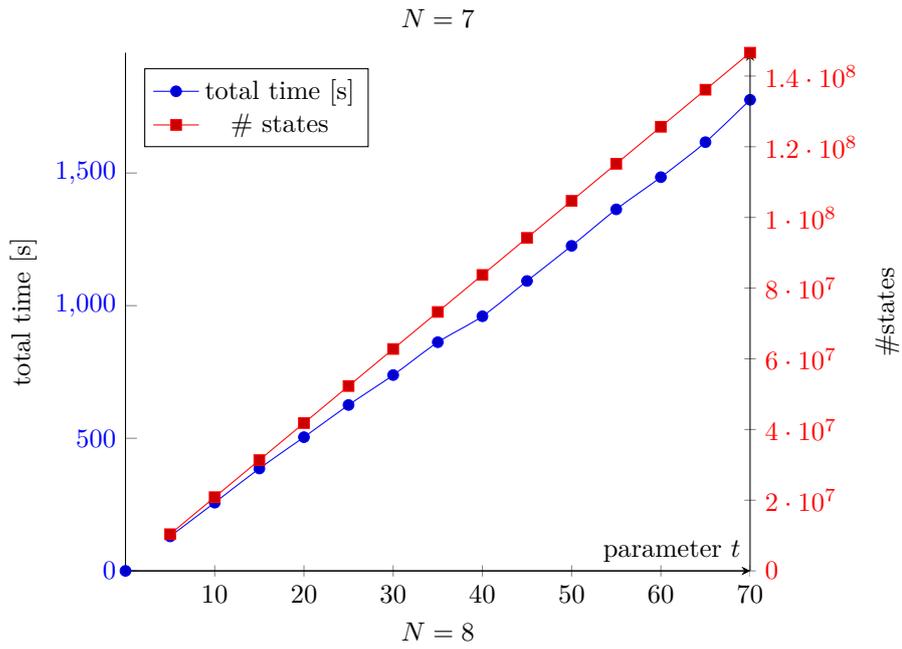
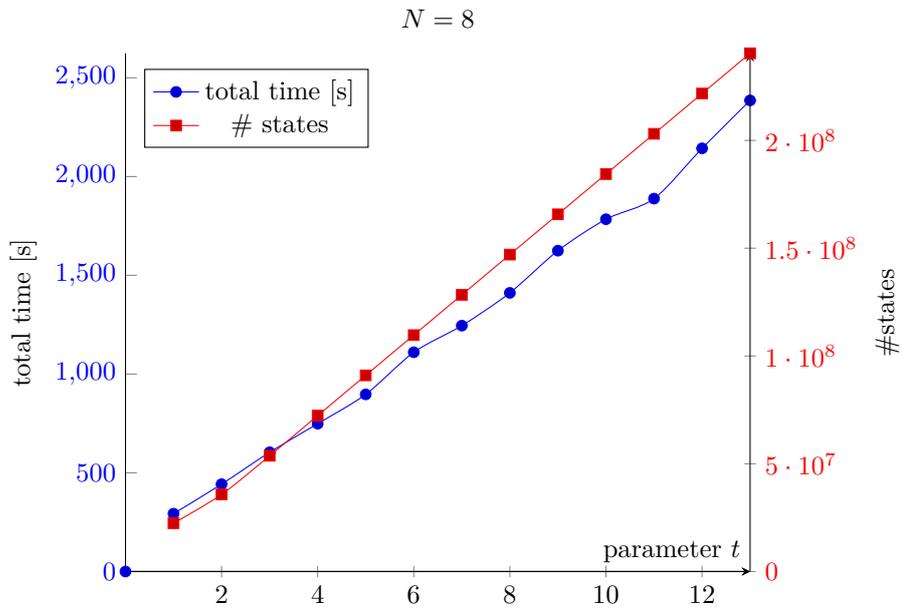
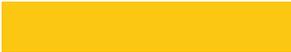
\end{document}